\documentclass[sigconf, nonacm]{aamas}

\usepackage{balance} 
\usepackage{algorithm}
\usepackage{algorithmic}

\makeatletter
\gdef\@copyrightpermission{
  \begin{minipage}{0.2\columnwidth}
   \href{https://creativecommons.org/licenses/by/4.0/}{\includegraphics[width=0.90\textwidth]{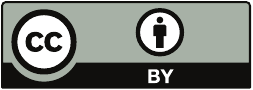}}
  \end{minipage}\hfill
  \begin{minipage}{0.8\columnwidth}
   \href{https://creativecommons.org/licenses/by/4.0/}{This work is licensed under a Creative Commons Attribution International 4.0 License.}
  \end{minipage}
  \vspace{5pt}
}
\makeatother

\setcopyright{ifaamas}
\acmConference[AAMAS '26]{Proc.\@ of the 25th International Conference
on Autonomous Agents and Multiagent Systems (AAMAS 2026)}{May 25 -- 29, 2026}
{Paphos, Cyprus}{C.~Amato, L.~Dennis, V.~Mascardi, J.~Thangarajah (eds.)}
\copyrightyear{2026}
\acmYear{2026}
\acmDOI{}
\acmPrice{}
\acmISBN{}

\title[AAMAS-2026 Formatting Instructions]{Altruism and Fair Objective in Mixed-Motive Markov games}

\author{Yao-hua Franck Xu}
\affiliation{
  \institution{Orange Labs Lannion}
  \city{Lannion}
  \country{France}}
\email{franck.xu@orange.com}

\author{Tayeb Lemlouma}
\affiliation{
  \institution{IRISA}
  \city{Lannion}
  \country{France}}
\email{tayeb.lemlouma@irisa.fr}

\author{Arnaud Braud}
\affiliation{
  \institution{Orange Labs Lannion}
  \city{Lannion}
  \country{France}}
\email{arnaud.braud@orange.com}

\author{Jean-Marie Bonnin}
\affiliation{
  \institution{IRISA}
  \city{Rennes}
  \country{France}}
\email{jean-marie.bonnin@irisa.fr}

\begin{abstract}
Cooperation is fundamental for society's viability, as it enables the emergence of structure within heterogeneous groups that seek collective well-being. However, individuals are inclined to defect in order to benefit from the group's cooperation without contributing the associated costs, thus leading to unfair situations. In game theory, social dilemmas entail this dichotomy between individual interest and collective outcome. The most dominant approach to multi-agent cooperation is the utilitarian welfare which can produce efficient highly inequitable outcomes.

This paper proposes a novel framework to foster fairer cooperation by replacing the standard utilitarian objective with Proportional Fairness. We introduce a fair altruistic utility for each agent, defined on the individual log-payoff space and derive the analytical conditions required to ensure cooperation in classic social dilemmas. We then extend this framework to sequential settings by defining a Fair Markov Game and deriving novel fair Actor-Critic algorithms to learn fair policies. Finally, we evaluate our method in various social dilemma environments.
\end{abstract}


\newcommand{\BibTeX}{\rm B\kern-.05em{\sc i\kern-.025em b}\kern-.08em\TeX}

\begin{document}


\pagestyle{fancy}
\fancyhead{}


\maketitle 

\section{Introduction}

Cooperation is a fundamental mechanism that underpins the stability and success of nearly all multi-agent systems, from biological ecosystems to complex human societies and artificial intelligence. It is the process by which independent agents coordinate their actions to pursue shared objectives, thus facilitating the emergence of complex structures and achieving a collective well-being that would be unreachable for any single agent acting alone. The capacity for effective cooperation is not merely advantageous; it is often crucial for survival, advancement, and the efficient use of resources.

Despite its critical importance, sustained cooperation is notoriously difficult to achieve. The primary obstacle is the inherent tension between individual self-interest and the collective good—a conflict formally studied in game theory as a social dilemma~\cite{social_dilemma}. In these scenarii, rational agents are frequently incentivized to defect from the group effort to pursue personal gain. This individual rationality can paradoxically culminate in collective failure, a phenomenon often described as a \textit{tragedy of the commons}, where the pursuit of self-interest leads to profoundly unfair and suboptimal outcomes for everyone~\cite{tragedy_of_common}.

Conceptually, the opposite to the selfish defection that drives these dilemmas is altruism—the principle of acting for the benefit of others. In the field of multi-agent learning~\cite{sutton_reinforcement_2014, marl-book}, this altruistic motivation is commonly formalized through a utilitarian social welfare objective, which seeks to maximize the (weighted) sum of all agents' utilities ~\cite{selfishness_SR, weighted_SR, cooperation}. By focusing on the greatest total good, this approach encourages agents to undertake actions that, while not always optimal for themselves, produce the most efficient outcomes for the group as a whole, thereby offering a direct method for resolving the dilemma. 

However, while the utilitarian approach can lead to efficient solutions, its exclusive focus on total welfare often ignores how rewards and sacrifices are distributed, frequently resulting in highly inequitable outcomes. For instance, a strict utilitarian calculus might conclude that a healthy individual should donate a kidney to a stranger, as the recipient's immense gain (a saved life) outweighs the donor's loss. Although total welfare is maximized, this outcome imposes an extreme and intuitively unfair sacrifice on one person. This highlights how pure, uncalibrated altruism can conflict with the principles of fairness necessary for stable cooperation. True, sustainable cooperation, therefore, requires more than just promoting the collective good; it demands a mechanism to ensure that the costs and benefits of working together are distributed fairly among all participants.

Fairness was first studied in resource allocation problems~\cite{nsw, pf_intro, PF_study}. In Multi-Agent Reinforcement Learning (MARL), fairness is merely emerging as an active research area and is usually treated alongside cooperation. The popular approaches to promote fairness in MARL use reward shaping~\cite{reward_shaping_0, reward_shaping_1, reward_shaping_2, reward_shaping_3} or role-based reward shaping~\cite{role_reward_shaping}. However, these methods are often environment-dependent, reward instantaneous fair behavior, and neglect the fairness on the full trajectory. While related works ~\cite{simir_to_us_but_not_deep, similar_to_us_but_no_policy_gradient} address similar problems, these studies are limited to finite trajectories and lack a deep analysis of the policy gradient.

To address these limitations and the need for a fairer approach, this paper introduces a novel framework designed to explicitly foster more equitable cooperation with heterogeneous agents. We move beyond simple reward maximization and propose a new objective based on Proportional Fairness, a concept that effectively balances efficiency with equity. 

Our contributions are threefold. First, we introduce a fair altruistic utility for each agent, defined on the individual log-payoff space, and derive the analytical conditions needed to ensure cooperation in classic social dilemmas. Second, we extend this framework to sequential settings by defining a Fair Altruistic Markov Game with infinite horizon. Finally, based on this new game formulation, we derive novel fair Actor-Critic algorithms designed to learn policies that achieve both high collective rewards and equitable outcomes, and we validate our method in various social dilemma environments.

\section{Preliminary}
\subsection{Notation}
We denote the set of real numbers by $\mathbb{R}$ and positive real numbers by $\mathbb{R}_+$. We denote $d$-dimensional real-numbered vector spaces as $\mathbb{R}^d$. By $\mathcal{P}(A)$, we denote the power set of a set $A$. For any probability measure $\mathbb{P}$, we denote the probability of an event $X=x$ as $\mathbb{P}(X=x)$. Similarly, we write $\mathbb{P}(X=x \,|\, Y=y)$ for conditional probabilities. When it is clear, which random variable is referred to, we regularly omit it to shorten notation, i.e.\ $\mathbb{P}(X=x) = \mathbb{P}(x)$.

We denote that a random variable $X$ follows a probability distribution $p$ by $X \sim p$. For any random variable $X \sim p$, we denote by $\mathbb{E}_{X\sim p}[X]$ its expectation. We denote the set of probability distributions over some measurable space $\mathcal{A}$ as $\Delta(\mathcal{A})$. We write $|A|$ for the cardinality of a finite set $A$. We use the vector notation $\vec{.}$ to refer to any element in a cartesian product, i.e for any sets $A_1,\,\dots,\,A_n$, $\vec{a}=(a_1,\,\dots,\, a_n)\in\bigtimes_{i=1}^n A_i$.

\subsection{Normal-form games}
We define a normal-form game $G$ by the tuple $\langle N, \{S_i\}_{i \in N}, \{u_i\}_{i \in N} \rangle$, where $N = \{1, 2, \ldots, n\}$ is the set of players, $S_i$ is the set of pure strategies available to player $i$.
Let $\mathcal{G}_{N, S}$ be the set of all normal form games with a fixed set of players $N = \{1, \ldots, n\}$ and fixed sets of pure strategies $S=(S_i)_{i \in N}$. $p_i: S \rightarrow \mathbb{R}$ is the payoff function for player $i$, which assigns a real-valued utility to each strategy profile $s \in S$. In this work, we assume the utility to be bounded.

\paragraph{Nash equilibrium}
In game theory, a fundamental solution concept is the Nash equilibrium where no player can benefit by changing their strategy while the other players keep theirs unchanged. Formally, a strategy profile $s^* = (s_1^*, s_2^*, \dots, s_n^*)$ is a Nash equilibrium if, for every player $i \in N$ and every alternative strategy $s_i \in S_i$:
\begin{equation*}
    p_i(s_i^*, s_{-i}^*) \geq p_i(s_i, s_{-i}^*).
\end{equation*}

\subsection{Altruistic games}
Altruism in game theory is often framed through social welfare functions (SWF), which aggregate individual utilities into collective wealth. While utilitarian SWFs maximize overall efficiency, they can permit extreme inequality. In contrast, "fair" SWFs satisfy axioms of efficiency (Pareto improvement), equity (Pigou-Dalton), and anonymity to ensure equitable distribution \cite{10.7551/mitpress/2954.001.0001}.
 
Following intuition, it is the smallest proportion of social welfare that needs to be offered to each player to make the socially optimal outcome a stable Nash equilibrium. In other terms, it is a measure of how much incentive is needed to make a selfish person act cooperatively. It's a way to measure the tension between an individual's self-interest and the collective good. This is the $\alpha$-selfish game, the most studied altruistic game~\cite{selfishness_SR}.

In this work, we adopt a slightly different definition~\cite{alpha_altruistic}.
\begin{definition}{($\alpha$-altruistic extension of a Normal-form Game)}
Given any normal-form game $G = \langle N, \allowbreak \{S_i\}_{i \in N}, \{p_i\}_{i \in N}\rangle$, we can induce an altruistic game $G(\alpha) \doteq \{N, \{S_i\}_{i \in N}, \{u_i\}_{i \in N}\}$ where, $u_i(s) \doteq (1-\alpha)p_i(s) + \alpha SW(s)$. The altruism level of a strategic game $G$ is then defined as:
\begin{equation}
    \alpha_G = \underset{\alpha}{\inf}\{\alpha \in \mathbb{R}_+ | G(\alpha) \text{ is } \alpha\text{-altruistic}\}
\end{equation}
where, $G$ is $\alpha$-altruistic if, for some $\alpha \ge 0$, a pure Nash equilibrium of $G(\alpha)$ is a social optimum of $G$.
    
\end{definition}

\subsection{Social dilemmas}
Social dilemmas are a class of normal-form games where players create a collectively worse outcome for the entire group when acting in their own self-interest. It highlights the dualism between individual preferences and the collective good. 

More formally, social dilemmas are symmetrical two-players normal-form games with payoff matrices similar as described in Table 1.

\begin{table}[h!]
\centering
\caption{Outcome categories in the payoff matrix}
\begin{tabular}{|c|c|c|}
\hline
 & C & D \\ 
\hline
C & $R, R$ & $S, T$ \\ 
\hline
D & $T, S$ & $P, P$ \\ 
\hline
\end{tabular}
\end{table}

where $R$ is the reward for mutual cooperation, $P$ the Punishment for mutual defection, $S$ the sucker's payoff for cooperating while the other defects and $T$ the temptation to defect while the other cooperates.

A situation is defined as a social dilemma if the payoffs satisfy the following preference inequalities.

\begin{table}[h!]
\centering
\caption{Set of inequalities that define social dilemmas.}
\renewcommand{\arraystretch}{1.2}
\begin{tabular}{p{1.5cm} p{6.2cm}}
\hline
\textbf{Inequality} & \textbf{Preference} \\
\hline
$R>P$& The individual prefers mutual cooperation (C, C) to mutual defection (D, D). \\
$R > S$& The individual prefers mutual cooperation to unilateral cooperation (C, D). \\
$2R > T + S$& The group prefers mutual cooperation to unilateral cooperation/defection. \\
$T > R$ or $P > S$ & The individual prefers unilateral defection (D, C) to mutual cooperation (aka, greed) or the individual prefers mutual defection to unilateral cooperation (aka fear). \\
\hline
\end{tabular}
\end{table}

\subsection{Proportional Fairness}
The resource allocation problem is the problem of allocating limited resources to competing tasks or individuals while maximizing an objective function -- for instances, performance, equity, or costs. The Proportional fairness (PF) objective was introduced to provide a compromise between efficiency and some kind of fairness, and was firstly introduced in telecommunication for rate control~\cite{PF_study, pf_intro, game_theory_harm_ratio}. 

\textbf{Formal Definition}: Let N be the set of agents, $X\subseteq\mathbb{R}_+^N$ the set of all feasible allocations. Let us denote by $u_i: X\rightarrow\mathbb{R}_+$ the utility of agent $i$. A solution $x^*\in X$ is proportionally fair if for any $x\in X$, the sum of proportional variation is negative:
\begin{equation*}
    \sum_{i=1}^N \frac{u_i(x)-u_i(x^*)}{u_i(x^*)}\leq 0
\end{equation*}
Proportional Fairness has been widely studied in the literature. In particular, it has been proven that a proportional fair allocation $x^*$ is a solution of the following optimization problem:
\begin{equation*}
    \max_{x\in X} \, \sum_{i=1}^N \log u_i(x)
\end{equation*}
Proportional fairness is directly linked to the Nash welfare function~\cite{nsw_def, nsw}, which seeks to maximize the product of agents' utilities, as maximizing the product of utilities is equivalent to maximize the sum of their logarithms in some sense.

\subsection{Markov game}
For modeling sequential decision-making in multi-agent settings, the standard framework is the Markov Game (also known as a Stochastic Game).

A Markov game $\mathcal{M}$ is defined by the following tuple: $(N, \mathcal{S}, \allowbreak  \{\mathcal{A}_i\}_{i \in N},   P, \rho_0, \{r_i\}_{i \in N}, \gamma)$, with $N$ the finite set of $n$ players, $\mathcal{S}$ the space of all possible states of the environment, $\{\mathcal{A}_i\}_{i\in N}$ the individual action set for each player, $P(s',\vec{a}, s) = \mathbb{P}(s'\,|\, s, \vec{a})$ the transition probability kernel, $\rho_0\in\Delta(\mathcal{S})$ the initial state distribution, $\{r_i\}_{i\in N}$ the individual reward for each player and $\gamma \in [0, 1)$ the discount factor.
In this work, we assume the state and action spaces to be finite and the reward functions to be positive and bounded by $R_{\max}$.

\section{Methodology}
In this section, we extend the $\alpha$-altruistic Game by rescaling the payoff function to use the Proportional Fairness function. Then, we study the effect of such a transformation on simple games, specifically two-player social dilemmas.

\subsection{Redefining Altruistic Games}
The former definition of the $\alpha$-altruistic game is not directly compatible with the proportional fairness function because summing individual rewards with log-rewards is not a homogeneous operation. To solve this, we apply a function $F_i$ to selfish rewards so they operate on the same scale as the social welfare function. For instance, we use $F_i=\log$ for PF Additionally, the logarithm function restricts the considered values to strictly positive real values. Hence our alternative definition,

\begin{definition}{($\alpha$-altruistic extension of a Normal-form Game)}
Given any normal-form game $G = \langle N, \allowbreak \{S_i\}_{i \in N}, \{p_i\}_{i \in N}\rangle$, we can induce a fair altruistic game $G(\alpha) \doteq \{N, \{S_i\}_{i \in N}, \{u_i\}_{i \in N}\}$ where, $u_i(s) \doteq (1-\alpha)F_i (p_i(s)-m_p) + \alpha SW(s)$, where $m_p=\inf_{i,s}\left[u_i(s)\right]$ is the min payoff a player can get over all possible strategy (which is assumed to exist) and $F_i:\mathbb{R\rightarrow\mathbb{R}}$ is an increasing function applied on player's payoff. The altruism level of a strategic game $G$ is then defined as:
\begin{equation}
    \alpha_G = \underset{\alpha}{\inf}\{\alpha \in \mathbb{R}_+ | G \text{ is } \alpha\text{-altruism}\}
\end{equation}
where, $G$ is $\alpha$-altruism if, for some $\alpha \ge 0$, a pure Nash equilibrium of $G(\alpha)$ is a social optimum of $G$.
\end{definition}

In the rest of this work, we set $F_i=\log$ for all players $i$.

\subsection{Analysis of Social Dilemmas}
Let $G \in\mathcal{G}_{N,S}$ be a normal form game where $N$ is the set of players, $S$ is the set of strategy profiles and $R$ is the payoff function, where $R_i(s)$ is the payoff to player $i$. We assume all payoffs are strictly positive, i.e., $R_i(s) > 0$ for all $i,\,s$.

We define a transformed game $G'(\alpha) = \langle N, S, R' \rangle$, where the transformed payoff function $R'$ is defined as:
\begin{align*}
    R'_i &= (1-\alpha)\log R_i+ \alpha \sum_{j\in N} \log R_j =\log R_i+\alpha  \sum_{j\neq i} \log R_j
\end{align*}

We now apply this framework to derive the threshold for $\alpha$ in the Prisoner's Dilemma, the Stag Hunt, and the Game of Chicken.

For the case of 2-player games ($N=2$), the rule for Player 1 and Player 2 simplifies to
$$R'_1 = \log R_1 + \alpha\log R_2,$$
$$R'_2 = \log R_2 + \alpha\log R_1.$$

The parameter $\alpha$ represents the degree of payoff interdependence. If $\alpha=0$, the game is unchanged ($R'_i = R_i$). As $\alpha$ increases, a player's transformed payoff becomes increasingly influenced by the other player's original payoff.

Our analysis will identify the conditions on $\alpha$ for the mutual cooperation outcome $(C, C)$ to be an equilibrium in $G'(\alpha)$. For this to be true, no player must have a profitable unilateral deviation. For a symmetric 2-player game, we only need to ensure that Player 1 does not gain by defecting, assuming Player 2 cooperates. This condition is
$$R'_1(C, C) \ge R'_1(D, C).$$

\begin{theorem}
    The altruism level of social dilemma is 
    \begin{equation}
        \alpha_G=
        \begin{cases}
            0 & \text{if } T \le R, \\
            \frac{\log T-\log R}{\log R-\log S} & \text{if } T > R.
        \end{cases}
    \end{equation}
\end{theorem}

\begin{proof}
    In the transformed game, the payoff becomes 
    \begin{align*}
        R'_2(C, C)=R'_1(C, C) &= (1+\alpha)\log R \\
        R'_1(D, C) &= \log T+\alpha \log S \\
        R'_2(D, C) &= \log S+\alpha \log T
    \end{align*}
    To make $(C, C)$ a Nash Equilibrium, the condition $R'_1(C, C) \ge R'_1(D, C)$ must hold.
    
    \paragraph{The Prisoner's Dilemma}
    The Prisoner's Dilemma is defined by the preference order $T > R > P > S$. The dilemma is that mutual defection $(D, D)$ is the unique Nash Equilibrium, despite mutual cooperation $(C, C)$ being a Pareto-superior outcome.
    
    The condition becomes $(1+\alpha)\log R \ge \log T+\alpha \log S$. Solving this inequality for $\alpha$ yields to the following condition 
    \begin{equation*}
        \alpha > \frac{\log T - \log R}{\log R -\log S}
    \end{equation*}

    \paragraph{The Stag Hunt}
    The Stag Hunt is defined by the preference order $R > P > S$ and $T\leq R$. This game has two pure-strategy equilibrium: (C, C) and (D, D). The cooperative outcome is already an equilibrium. The analysis thus investigates whether the transformation preserves this equilibrium.
    
    The condition is identical to the one above: $R'_1(C, C) \ge R'_1(D, C)$. The resulting inequality is also $(1+\alpha)\log R \ge \log T+\alpha \log S$. However, the implications are different. In a Stag Hunt, $T \leq R$, so $\log T - \log R$ is negative, while $\log R-\log S$ is positive. The ratio $\frac{\log T - \log R}{\log R -\log S}$ is therefore negative. As $\alpha$ is always positive, the condition is met for any $\alpha\in[0, 1]$. This confirms that the transformation never change the cooperative equilibrium.

    \paragraph{The Game of Chicken}
    The Game of Chicken is defined by the preference order $T > R > S \geq P$. Mutual cooperation $(C, C)$ is not an equilibrium because each player is tempted to defect to achieve the highest payoff.

    To make $(C, C)$ a Nash Equilibrium, the following condition must be met
    \begin{equation*}
        \alpha > \frac{\log T - \log R}{\log R -\log S}
    \end{equation*}

    Therefore, cooperation becomes a Nash Equilibrium whenever the following condition is met 
    \begin{equation*}
        \alpha=
        \begin{cases}
            0 & \text{if } T \le R, \\
            \frac{\log T-\log R}{\log R-\log S} & \text{if } T > R.
        \end{cases}
    \end{equation*}

    To ensure consistency with the definition of altruism level $\alpha\leq 1$, we require $TS \le R^2$, which holds for standard payoff matrices. 

\end{proof}

In each case where cooperation is not initially stable (Prisoner's Dilemma, Chicken), the required threshold for $\alpha$ is a ratio of logarithmic payoffs. The numerator, $\log(T/R)$, represents the "temptation" to defect from the cooperative state. The denominator reflects the consequences of defection. The transformation effectively allows players to internalize the externalities of their actions, and the threshold for $\alpha$ quantifies the precise degree of internalization needed to make cooperation individually rational.

The use of proportional fairness combined with altruism is evident in traditional static social dilemmas. However, the simplicity of Normal-form games is a limitation. Indeed, real-world phenomena rarely occur in a single step, but rather evolve over time. To accurately model such dynamic systems, the analysis must move beyond such formalism to a framework like the Markov Game.

\section{Fair Altruistic Markov Game}
In this section, we propose and analyze a novel objective function integrating Proportional Fairness and Altruism within the Markov Game framework and derive policy gradient theorems\footnote{We provide more detailed proofs in Appendix \ref{app:appendix}.}.

\subsection{Value functions}
To extend our approach to Markov Games, we must establish a dynamic analogy to the players' static payoff function for a given policy profile $\vec{\pi}$. This can be achieved using the individual state-value function which represents the expected cumulated discounted rewards for a single player over all possible infinite trajectories.

One challenge in defining the state-value function is to construct a single, consistent probability measure, $\mathbb{P}_{\vec{\pi}}$, over the entire infinite set of trajectories, $\mathcal{T} = (\mathcal{S} \times \mathcal{A})^{\infty}$. We do this by defining the probability of any finite sequence and then extending it to the infinite case using the Ionescu-Tulcea Theorem.

\noindent First of all, let define the discounted return $R_i$ for each agent $i\in N$ and for a given trajectory $\tau=(s_0, \, \vec{a_0}, \, s_1,\, \vec{a_1},\,\dots)\in \mathcal{T}$
\begin{equation*}
    R_{i}(\tau)=\sum_{t=0}^{\infty} \gamma^t\;r_{i}(s_{t}, \vec{a}_{t}).
\end{equation*}
As the reward function is bounded, the discounted return converges absolutely.

The state-value function of agent $i$ is the expected return with respect to $\mathbb{P}_{\vec{\pi}}$ given the initial state $s_0\in\mathcal{S}$,

\begin{equation*}
    V_i^{\vec{\pi}}(s) = \underset{\tau \sim \mathbb{P}_{\vec{\pi}}}{\ \mathbb{E}} \left[R_i(\tau) \;\middle|\; s_0=s\right].
\end{equation*}

Similarly, the action-value function of agent $i$ is defined as the expected discounted return given the initial state and the first joint action, 
\begin{equation*}
    Q_i^{\vec{\pi}}(s, \vec{a}) = \underset{\tau \sim \mathbb{P}_{\vec{\pi}}}{\ \mathbb{E}} \left[R_i(\tau) \;\middle|\; s_0=s,\,  \vec{a}_0=\vec{a}\right]
\end{equation*}

In a Markov Game, the Bellman operator is a fundamental concept used in multi-agent reinforcement learning to define the value of a state for a specific agent.
\begin{definition}
In a Markov Game with a set of agents $N$ and joint stationary policies $\vec{\pi}$, the Bellman operator for agent $i$, denoted as $T_i^{\vec{\pi}}$, is applied to any state-value function $V_i(s)$ and is defined as:
$$(T_i^{\vec{\pi}} V_i)(s) = \sum_{\vec{a} \in \mathcal{A}} \pi(\vec{a}\,|\, s) \left( r_i(s, \vec{a}) + \gamma \sum_{s' \in \mathcal{S}} \mathbb{P}(s' | s, \vec{a}) V_i(s') \right).$$
\end{definition}

We can then write the Bellmann theorem which gives a key equation verified by $ V_i^{\vec{\pi}}$
\begin{theorem}
    For any joint policy $\vec{\pi}$, and any state-value function associated with $\vec{\pi}$, the following equation called Bellman equation holds 
    \begin{equation*}
        T_i^{\vec{\pi}} V_i^{\vec{\pi}} = V_i^{\vec{\pi}}
    \end{equation*}
\end{theorem}

\subsection{Fair objective}
From the traditional Markov Game framework, we can induce a new objective that balances fairness and efficiency and extends the altruistic Normal-form game to the Markov game. 
\begin{definition}
    For each agent $i$, we define the proportional fair state value function :
    \begin{equation*}
    V^{\vec{\pi}, \text{ Prop}}(s) =\sum_{j=1}^n \log V_{j}^{\vec{\pi}}(s).
\end{equation*}
\end{definition}
The state value functions play the role of the payoff function and the proportional fairness function is applied on the state value function.
Therefore, we can define a new objective for each agent depending on the altruism level:
\begin{definition}
    The objective for each agent is to maximize
    \begin{align*}
        J_i(\vec{\pi})&=\underset{s_0\sim\rho_0}{\mathbb{E}}\left[(1-\alpha)\log V_i^{\vec{\pi}}(s_0)+\alpha  V^{\vec{\pi}, \text{ Prop}}(s_0)\right]\\
        &=\underset{s_0\sim\rho_0}{\mathbb{E}}\left[\sum_{j=1}^N c_i(j)\log V_j^{\vec{\pi}}(s_0)\right]
    \end{align*}
    where 
     \begin{equation*}
        c_i(j)=\left\{ 
        \def\arraystretch{1.5}
        \begin{array}{ll}
            1&\text{ if } j=i,\\
            \alpha &\text{ if } j\neq i.
    \end{array} 
    \right.
    \end{equation*}
\end{definition}

\subsection{Policy gradient}
Policy gradient methods are a class of reinforcement learning algorithms that directly learn and optimize a parameterized policy, $\pi_\theta(a\,|\,s)$, without consulting a value function. The core principle is to adjust the policy's parameters, $\theta$, to increase the probability of actions that yield high cumulative rewards and decrease the probability of actions that result in low rewards.

Let us suppose that agent $i$'s policy $\pi_i$ is parameterized by $\theta_i$ such that $\pi_{\theta_i}$ is differentiable with respect to $\theta_i$.

Therefore, the goal of agent $i$ is to maximize the objective $J_i(\vec{\theta})=\mathbb{E}_{s_0\sim\rho_0}[\,\sum_{j=1}^Nc_i(j)\log V_j^{\vec{\theta}}(s_0)\,]$. We can demonstrate the differentiability of this objective function with respect to any parameter $\theta$ as summarized in the following theorem.
\begin{theorem}
    (\text{Policy gradient theorem.}) Let $\mathcal{M}$ be a Markov game and $J_i(\vec{\theta})=\mathbb{E}_{s_0\sim\rho_0}\left[\sum_{j=1}^N c_i(j)V_j^{\vec{\theta}}(s_0)\right]$ the objective defined previously and for all agent $i$, let us suppose $Q^{\vec{\pi}}_i$ is differentiable with respect to $\theta_i$ under mild conditions. The gradient of $J_i$ with respect to $\theta_i$ is given by  
    \begin{equation*}
        \nabla_{\theta_i}J_i(\vec{\theta})=\underset{\tau\sim\mathbb{P}^{\vec{\theta}}}{\mathbb{E}}\left[\sum_{t=0}^\infty \sum_{j=1}^Nc_i(j)\frac{Q^{\vec{\theta}}_j(s_t, \, \vec{a}_t)}{V^{\vec{\theta}}_j(s_0)}\nabla_{\theta_i}\log\pi_{\theta_i}(a_{i,\, t} \;|\; s_t)\right]
    \end{equation*}
\end{theorem}

\begin{proof}
Let us derive the objective for agent $i$ using the chain rule,
\begin{align*}
    \nabla_{\theta_i}J_i(\vec{\theta})
    &=\underset{s_0\sim\rho_0}{\mathbb{E}}\left[\sum_{j=1}^N c_i(j)\frac{\nabla_{\theta_i}V_j^{\vec{\theta}}(s_0)}{V_j^{\vec{\theta}}(s_0)}\right]
\end{align*}

Then, we need to find the gradient $\nabla_{\theta_i}V_j^{\vec{\theta}}(s_0)$. To do so, we start from the Bellman equation, differentiate it and define a contraction mapping on the state value function gradient space and apply the Banach-Picard Fixed-Point Theorem to get the formula of the gradient.

\textbf{1. Differentiate the Bellman equation.}
$r_i$, $P$, and policies $\pi_{\theta_j}$ for $j \neq i$ do not depend on $\theta_i$.
Then, the Bellman equation can be differentiate using the product rule into a recursive equation:
\begin{align*}
    \nabla_{\theta_i}V_j^{\vec{\theta}}(s) &=\sum_{\vec{a} \in \mathcal{A}} \nabla_{\theta_i}\pi_{\vec{\theta}}(\vec{a}\,|\, s) \cdot\left( r_j(s, \vec{a}) + \gamma \sum_{s' \in \mathcal{S}} \mathbb{P}(s' | s, \vec{a}) V_j^{\vec{\theta}}(s') \right) \\
    &\quad\quad +\,  \pi_{\vec{\theta}}(\vec{a}\,|\, s)\cdot \gamma \sum_{s' \in \mathcal{S}} \mathbb{P}(s' | s, \vec{a}) \nabla_{\theta_i} V_j^{\vec{\theta}}(s')\\
\end{align*}

Let us define a function $G_{i, j}^{\vec{\theta}}(s)$ which represents the immediate gradient component at state $s$
\begin{align*}
    G_{i, j}^{\vec{\theta}}(s) &:= \sum_{\vec{a}\in\mathcal{A}} \pi_{\vec{\theta}}(\vec{a}|s) (\nabla_{\theta_i} \log \pi_{\theta_i}(a_i|s)) Q_j^{\vec{\theta}}(s, \vec{a}) \\
    &= \mathbb{E}_{\vec{a} \sim \vec{\pi}_\theta} [ \nabla_{\theta_i} \log \pi_{\theta_i}(a_i|s) Q_j^{\vec{\theta}}(s, \vec{a}) ]
\end{align*}

With this, our recursive equation can be simplified into

$$\nabla_{\theta_i} V_j^{\vec{\theta}}(s) = G_{i,j}^{\vec{\theta}}(s) + \gamma \sum_{\vec{a}\in\mathcal{A}} \pi_{\vec{\theta}}(\vec{a}|s) \sum_{s'} \mathbb{P}(s'|s, \vec{a}) \nabla_{\theta_i} V_j^{\vec{\theta}}(s')$$

\textbf{2. Define Bellman operator for the value function gradient}

This equation defines a fixed-point relationship. Let $g_j: S \to \mathbb{R}^{|\theta_j|}$ be a vector field of gradients. We can define the Bellman operator for the value function gradient, $\mathfrak{T}_{i, j}^{\vec{\theta}}$, as
$$(\mathfrak{T}_{i, j}^{\vec{\pi}} g_j)(s) := G_{i,j}^{\vec{\theta}}(s) + \gamma \mathbb{E}_{\vec{a} \sim \pi_{\vec{\theta}}(\cdot|s),\; s' \sim P(\cdot|s,\vec{a})} [g_j(s')]$$

The gradient we are looking for, $\nabla_{\theta_i} V_j^{\vec{\theta}}$, is a fixed point of this operator
$$\nabla_{\theta_i} V_j^{\vec{\theta}} = \mathfrak{T}_{i, j}^{\vec{\theta}} (\nabla_{\theta_i} V_j^{\vec{\theta}})$$

\textbf{3. Policy Gradient Theorem for Markov Games}

The previous operator $\mathfrak{T}_{i, j}^{\vec{\theta}}$ defined a contraction mapping on the state-value function gradient space. By the \textbf{Banach-Picard fixed-point theorem}, it has a unique fixed point. This proves that the gradient vector field $\nabla_{\theta_i} V_j^{\vec{\theta}}$ is unique. 
Moreover, the \textbf{Banach-Picard fixed-point theorem} also states that for any sequence $(g_n)_{n\in\mathbb{N}}$ where $g_{n+1}=\mathfrak{T}_{i, j}^{\vec{\theta}}(g_n)$ and $g_0:\mathcal{S}\rightarrow\mathbb{R}^{|\theta_j|}$, $(g_n)_{n\in\mathbb{N}}$ converges $\nabla_{\theta_i}V_j^{\vec{\theta}}$.

We can find the solution by unrolling the recursion and starting with $g_0=0$. As $g_n$ converges to $\nabla_{\theta_i} V_j^{\vec{\theta}}$, the following equality holds 
\begin{equation*}
    \nabla_{\theta_i} V_j^{\vec{\theta}}(s_0)=g_\infty(s_0)=\underset{\tau\sim \mathbb{P}^{\vec{\theta}}}{\mathbb{E}} \left[ \sum_{t=0}^{\infty} \gamma^t G_{i,j}^{\vec{\theta}}(s_t) \;\middle|\; s_0\right]
\end{equation*}

Substituting the definition of $G_{i,j}^{\vec{\theta}}(s_t)$:
$$\nabla_{\theta_i} V_j^{\vec{\pi}}(s_0) = \underset{\tau\sim \mathbb{P}^{\vec{\theta}}}{\mathbb{E}}\left[ \sum_{t=0}^{\infty} \gamma^t \nabla_{\theta_i} \log \pi_{\theta_i}(a_{i,t}|s_t) Q_j^{\vec{\theta}}(s_t, \vec{a}_t) \;\middle|\; s_0\right]$$

We can simplify this equation and inject the expression of $\nabla_{\theta_i} V_j^{\vec{\pi}}(s_0)$ into the former gradient, we get the \textbf{Fair Policy Gradient Theorem}

\begin{equation*}
    \nabla_{\theta_i}J_i(\vec{\theta})=\underset{\tau\sim \mathbb{P}^{\vec{\theta}}}{\mathbb{E}}\left[ \sum_{t=0}^{\infty} \gamma^t \nabla_{\theta_i} \log \pi_{\theta_i}(a_{i,t}|s_t) \left(\sum_{j=1}^N c_i(j)\frac{Q_j^{\vec{\theta}}(s_t, \vec{a}_t)}{V_j^{\vec{\theta}}(s_0)}\right)\right]. \\
\end{equation*}
\end{proof}


One common technique in reinforcement learning to stabilize learning is to add a baseline to the gradient. The same technique can be applied in our case. 
\begin{theorem}
    (Fair Altruistic Advantage Policy Gradient). 
    The gradient can be rewritten 
    \begin{equation*}
        \nabla_{\theta_i}J_i(\vec{\theta})=\underset{\tau\sim\mathbb{P}^{\vec{\pi}_\theta}}{\mathbb{E}}\left[\sum_{t=0}^\infty A^{F, \vec{\theta}}_i(s_t,\vec{a}_t, s_0)\nabla_{\theta_i}\log\pi_{\theta_i}(a_{i,\, t} \;|\; s_t)\right]
    \end{equation*}
    where $A^{F,\vec{\theta}}_i(s,\vec{a}, s')=\sum_{j=1}^N c_i(j)\frac{A^{\vec{\theta}}_j(s, \, \vec{a})}{V^{\vec{\theta}}_j(s')}$ is the fair altruistic advantage function of agent $i$.
\end{theorem}
\begin{proof}

The proof relies on showing that the expected value of the gradient term associated with any baseline function of $s_t$ is zero. Let us consider the expectation of this baseline term at a single timestep $t$ for a given state $s_t$. The expectation is over the joint action $\vec{a}_t$ drawn from the joint policy $\vec{\pi}(\cdot|s_t)$.

We want to show that for any $f:\mathcal{S}\rightarrow\mathbb{R}$,
\begin{equation*}
    \mathbb{E}_{\vec{a}_t \sim \pi_{\vec{\theta}}(\cdot|s_t)} \left[ \nabla_{\theta_i} \log \pi_{\theta_i}(a_{i,t}|s_t) f(s_t) \right] = 0
\end{equation*}

Notice that $f(s_t)$ does not depend on $\vec{a}_t$ so we can pull it out of the expectation,
\begin{equation*}
    f(s_t)\cdot \mathbb{E}_{\vec{a}_t \sim \pi_{\vec{\theta}}(\cdot|s_t)} \left[ \nabla_{\theta_i} \log \pi_{\theta_i}(a_{i,t}|s_t)\right].
\end{equation*}

The inner function in the expectation depends only on the action of agent $i$. Then, we can remove the other actions' influence from the expectation,
\begin{equation*}
    f(s_t)\cdot \mathbb{E}_{a_{i,t} \sim \pi_{\theta_i}(\cdot|s_t)} \left[ \nabla_{\theta_i} \log \pi_{\theta_i}(a_{i,t}|s_t)\right].
\end{equation*}

Let us expand the expectation
\begin{equation*}
    f(s_t)\sum_{a_{i,t} \in \mathcal{A}_i} \pi_{\theta_i}(a_{i,t}|s_t) \nabla_{\theta_i} \log \pi_{\theta_i}(a_{i,t}|s_t).
\end{equation*}

Using the log-derivative trick 
\begin{align*}
    \sum_{a_{i,t}} \pi_{\theta_i}(a_{i,t}|s_t) \nabla_{\theta_i} \log \pi_{\theta_i}(a_{i,t}|s_t) 
    &= \sum_{a_{i,t}} \nabla_{\theta_i} \pi_{\theta_i}(a_{i,t}|s_t)\\
    &= \nabla_{\theta_i} \sum_{a_{i,t}} \pi_{\theta_i}(a_{i,t}|s_t)\\
    &= \nabla_{\theta_i} (1)\\
    &=0.
\end{align*}

Since the expected value of the baseline term is zero, subtracting it from the term inside the expectation in the original policy gradient theorem does not change the total expectation. In this work, we use the baseline $f(s_t)=\sum_{j=1}^Nc_i(j)\frac{V_j^{\vec{\theta}}(s_t)}{V_j^{\vec{\theta}}(s_0)}$. We then get the advantage form of the fair policy gradient
\begin{align*}
    \nabla_{\theta_i}J_i(\vec{\theta})&=\underset{\tau\sim\mathbb{P}^{\vec{\theta}}}{\mathbb{E}}\left[\sum_{t=0}^\infty A^{F,\vec{\theta}}_i(s_t,\vec{a}_t, s_0)\nabla_{\theta_i}\log\pi_{\theta_i}(a_{i,\, t} \;|\; s_t)\right] \\
\end{align*}
\end{proof}

In particular, when $\alpha=1$, all agents try to maximize the same objective 
\begin{equation*}
    J_i(\vec{\theta})=\underset{s_0\sim\rho_0}{\mathbb{E}}\left[ V^{\vec{\theta}, \text{ Prop}}(s_0)\right]
\end{equation*}
and the policy gradient becomes 
\begin{equation*}
    \nabla_{\theta_i}J_i(\vec{\theta})=\underset{\tau\sim\mathbb{P}^{\vec{\pi}_\theta}}{\mathbb{E}}\left[\sum_{t=0}^\infty A^{F, \vec{\theta}}(s_t,\vec{a}_t, s_0)\nabla_{\theta_i}\log\pi_{\theta_i}(a_{i,\, t} \;|\; s_t)\right]
\end{equation*}
where $A^{F,\vec{\theta}}(s,\vec{a}, s')=\sum_{j=1}^N\frac{A^{\vec{\theta}}_j(s, \, \vec{a})}{V^{\vec{\theta}}_j(s')}$ is the fair advantage function.

From this policy gradient theorem, we can derive multi-agent Actor-Critic algorithms that learn fair policies. We try to validate our method by comparing measured fairness and efficiency of multi-agent systems in a well-known environment. 

\section{Experiments}
In this section, we present the experiments. In a nutshell, we study the effects of the fair objective on the agents' performances and the group's fairness under various altruism level in a popular mixed-motive game called 'CleanUp'.
\subsection{Environment}

The 'CleanUp' environment was introduced in \cite{hughes2018inequityaversionimprovescooperation} and is now part of  Melting Pot 2.0~\cite{melting_pot}. It is a commonly used multi-agent substrate designed to explore complex social dilemmas. It is built as a resource management game where multiple agents must cooperate to maintain a renewable resource pool. In 'CleanUp', agents have to harvest as many apples as they can. Agents only gain an immediate positive reward for harvesting one apple. 

\begin{figure}[h]
  \centering
  \includegraphics[width=0.8\linewidth]{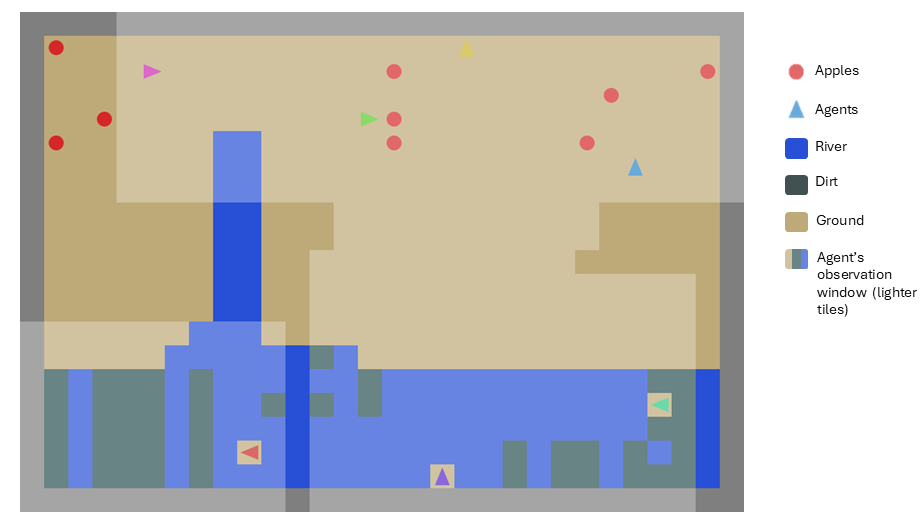}
  \caption{The CleanUp environment. Agents' observation window is represented in lighter tiles.}
  \label{fig:cleanup}
  \Description{A view of CleanUp environment. It consists on a closed environment with a river and a ground. Apples appears randomly based on how clean the river is.}
\end{figure}

\begin{figure*}[h]
  \centering
  \includegraphics[width=1\linewidth]{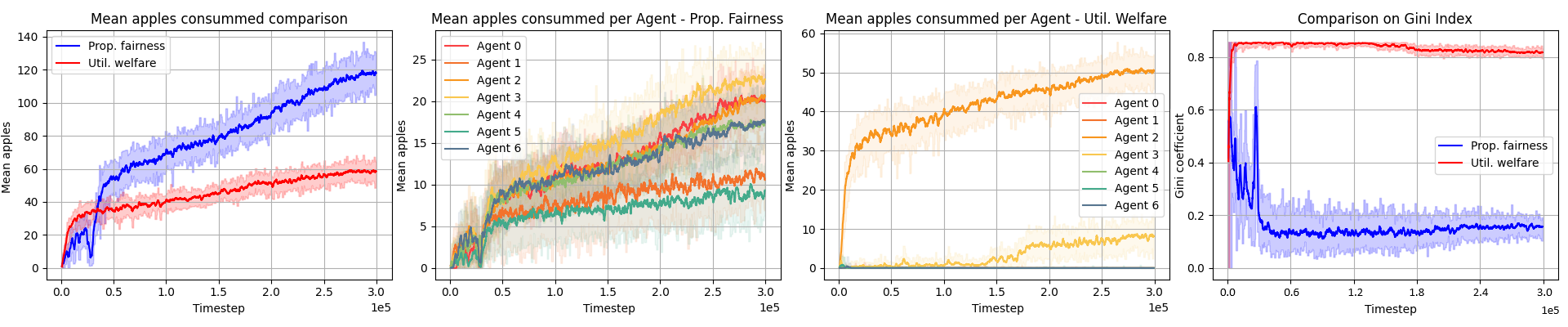}
  \caption{Comparison between Prop. Fairness and Util. Welfare objective trained with MAPPO in the fully cooperative setting ($\alpha=1$). Solid lines is obtained by averaging the metric over a rolling window of 50 runs. Shaded areas represent the min and the max of the respective metric over the same rolling window.  }
  \label{fig:prop_util}
  \Description{Comparison between Prop. Fairness and Util. Welfare objective trained with MAPPO in the fully cooperative setting ($\alpha=1$). There are four curves, comparison of Mean Apples consumed, comparison Mean Apples consumed per agent for both objectives and Gini index comparison. }
\end{figure*}

However, the growth rate of apples depends on how clean the river is. Whenever the pollution level becomes too high, the growth rate drops to zero and no additional apples can appear. To prevent this, agents have to clean the river and reduce pollution, benefiting all players. Thus, agents have to sacrifice short-term gains for long-term rewards as they cannot harvest apples while cleaning the river at the same time. This setup forces agents to navigate a classic social dilemma, balancing the selfish incentive of individual reward maximization (harvesting) against the collective necessity of long-term sustainability (cleaning), thereby testing their capacity for cooperation and pro-social behavior.

\subsection{Setup}
The environment is setup with partial observability. Agents can only view a part of the world consisting of a squared window of size $11\times 11$ pixels around their position. In this work, we instantiate any environment with 7 agents with random initial positions.

Our codebase\footnote{Code available at \url{https://github.com/AkuBrains/altruistic-fair-MARL/}} leverages the Jax-based 'CleanUp' environment provided by the SocialJax framework~\cite{social_jax}. We employ decentralized Actor-Critic algorithms, where each agent is modeled by its own set of neural networks, ensuring non-shared parameters (i.e., independent learning). This non-shared parameterization allows us to simulate a broader range of strategy profiles and better capture the full complexity of the CleanUp social dilemma. In particular, we use two convolutional layers to handle the spatial nature of observations followed by two fully connected layers. For each agent, the actor network represents the parameterized policy while the critic network estimates the individual state-value function.

We ran our training algorithms over 10 parallel environments. The episode length is set to 100 timesteps. We keep the episode length quite low on purpose to study the impact of the initial position on the learning process. Players spawned close to the apple area have an initial advantage over those closer to the river even with random policies. The total timestep is set to $3\times 10^5$. The learning rate is set to 0.001 for PPO-like algorithms~\cite{yu_surprising_2022} and 0.01 for A2C-like algorithms, with a decay schedule applied on both to reach a final value of 0.00001. In PPO-liked algorithms, we make advantage of the entropy regularization to improve exploration.

In our experiment, we study the performance of the following algorithms.
\paragraph{Fair MAA2C}
It is a straight forward application of the policy gradient theorem. Each agent's actor is updated using the individual Actor's loss given by $\mathbb{E}_t[A^{F, \vec{\theta}}_{i, t}\log\pi_{\theta_i}(a_{i,\, t} \;|\; s_t)]$.
\paragraph{Fair MAPPO} 
It is derived from MAPPO. We keep the core idea of MAPPO and we replace the classical advantage function with the fair advantage function in the PPO loss despite mathematical unaccuracies.

To assess the performance of our algorithms, we use two metrics. The first one evaluates the efficiency and is defined by the total number of apples consumed over an episode. The second one assesses the fairness among the agents and is defined by the Gini coefficient $Gini_t$ over the apples consumed, where 
\begin{equation*}
    Gini_t(c_0, \dots, c_n) = \frac{\sum_{i\in N}\sum_{j\in N}|c_i-c_j|}{2N\sum_{i\in N}c_i}
\end{equation*}
with $c_i$ being the number of apples consumed by agent $i$ in an episode.
The Gini coefficient navigates between $0.0$ and $1.0$. The lower it is, the fairer the distribution is, reaching $0.0$ for an even distribution. On the other hand, the higher it is, the more unequal the distribution is. 

We also study the fully cooperative setting in more depth and compare the Proportional Fairness (PF) objective and the Utilitarian Welfare (UW) objective.

\subsection{Results}


Figure \ref{fig:prop_util} presents a comparison between the PF objective and the UW one in a fully cooperative settings ($\alpha=1$). Surprisingly, the Prop. Fairness approach demonstrates superior overall efficiency. The total number of apples consumed by the group steadily increases, reaching approximately 120 by the end of the training. On the contrary, the Utilitarian driven method yields a significantly lower total apple consumption, leveled off near 40. 

After breaking down the overall consumption, all seven agents consume a relatively similar number of apples in the PF setting. Although there are minor variations, the trend lines for all agents are clustered together. No single agent is left behind, and all seem to contribute and benefit from the fair task which is reflected in the low Gini coefficient. In stark contrast, only two agents (2 and 3) learn to harvest apples while the others are cleaning the river or doing nothing in the UW setting. Thus, this outcomes in a very unequal situation characterized by a very high Gini index near $0.8$.

\begin{figure}[h]
  \centering
  \includegraphics[width=\linewidth]{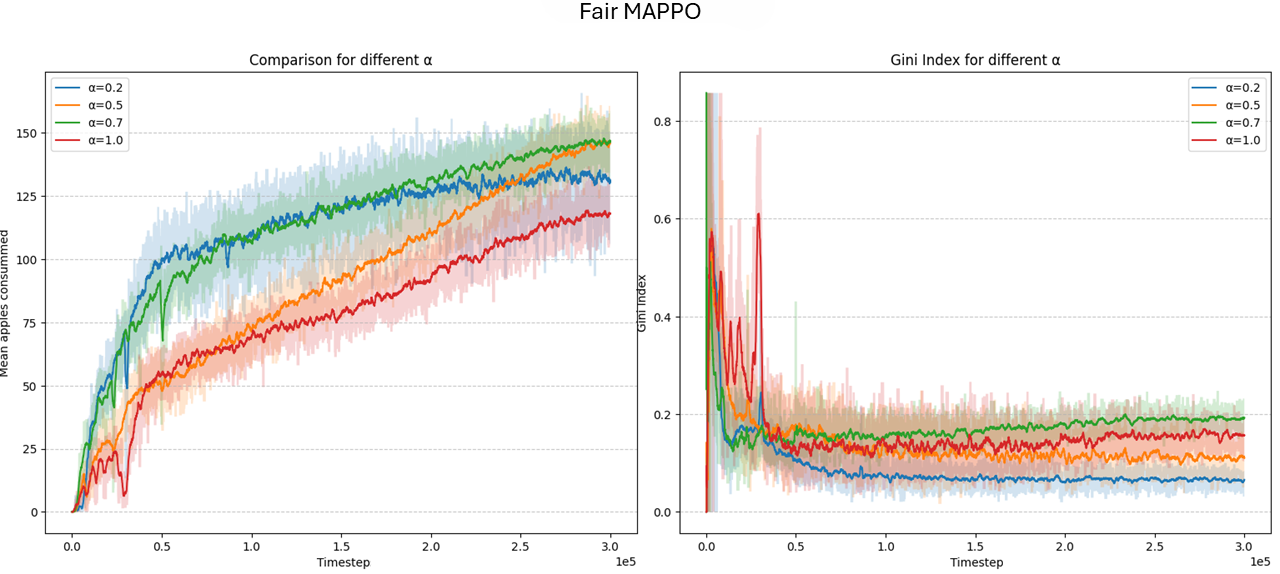}
  \caption{Performance of varying $\alpha$ in the CleanUp environment trained with Fair-MAPPO.}
  \label{fig:mappo}
  \Description{Results for Fair MAPPO. Compares the Mean Apples consumed and the Gini index during the training process for different altruism level.}
\end{figure}

In the Utilitarian settings, each agent benefits from others' reward signal independently of the action the agent takes. Therefore, the reward gained by the agent and its actions are not necessarily correlated. Agents are trapped into sub-optimal policies creating specialized role for each agent. When agents share the same parameters, they basically learn a single policy assuming natural cooperation and homogeneous behaviors. In the PF setting, the natural fairness brought by the PF objective encourages agents to explore further whenever an imbalance occurs. In the Fair Advantage function, the classical advantage function is normalized by the state-value function significantly reducing the weight of others' reward signal in the PF objective. 

Figure \ref{fig:mappo} shows the results for Fair MAPPO. All agent groups successfully learn to collect more apples over time, as indicated by the upward trend of all curves. Moreover, the Gini index of all experiments is quite low and all runs are clustered between $0.05$ and $0.2$. The group with $\alpha=0.7$ achieves the best overall performance, consuming the most apples by the end of the training. Against all odds, the worst group is the fully cooperative one achieving significantly lower scores in both overall performance and Gini index.

Intuitively, in the fully cooperative setting, the same phenomenon as in the UW setting exists but attenuated by the equity willingness of the PF objective. In contrast, when $\alpha<1$, agents are more inclined to follow their own objective and are prone harvesting apples autonomously. When $\alpha$ is too low ($\alpha=0.2$), the overall harvesting performance drops. In that case, agents are merely selfish meaning they are less inclined to coordinate themselves to clean the river.

\begin{figure}[h]
  \centering
  \includegraphics[width=\linewidth]{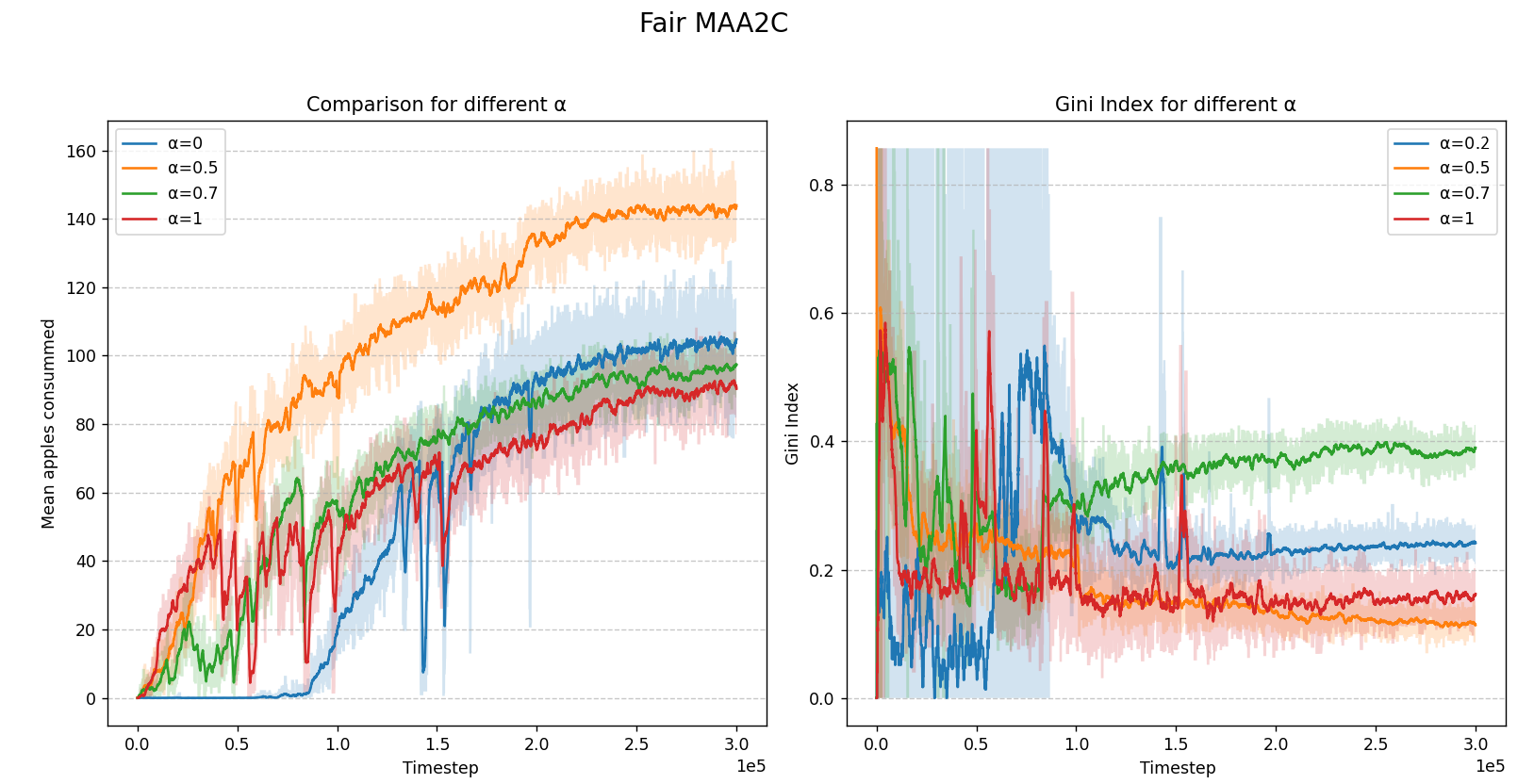}
  \caption{Performance of varying $\alpha$ in the CleanUp environment trained with Fair-MAA2C.}
  \label{fig:maa2c}
  \Description{Results for Fair MAA2C. Compares the Mean Apples consumed and the Gini index during the training process for different altruism level.}
\end{figure}

Figure \ref{fig:maa2c} shows the results for Fair MAA2C. All agent groups successfully learn to collect more apples over time, as indicated by the upward trend of all curves. However, the training is very noisy as we can notice from the several steep drops of the average apples consumed during the training. It is well known that MAA2C training process is very unstable and our Fair MAA2C exhibits the same behavior making the results difficult to interpret. Nonetheless, the best utilitarian performance and Gini index is obtained for $\alpha=0.5$ correlating Fair MAPPO results. In addition, Fair MAA2C seems to achieve lower mean apples consumed compared to its PPO counterpart except for $\alpha=0.5$ probably due to instability during the learning process. 


In short, selfish behavior can lead to a more equal outcome, but only because everyone is competing for the same individual rewards, potentially at the expense of the group's overall success. This highlights the fundamental trade-off between individual selfishness and collective goals. We also demonstrate that the fairness promoted by the PF objective enables sustainable cooperation, whereas a utilitarian approach fails to do so, particularly in heterogeneous groups.

\subsection{Discussions and Next steps}
Our algorithms still need to be improved regarding theoretical proofs. In particular, there is no convergence proof to any kind of equilibrium if it exists. Moreover, replacing the traditional Advantage function with the Fair Advantage function is inconsistent with the theory as we didn't prove the monotonic improvement theorem of Trust Region Policy Optimization in that case~\cite{trpo, HARL}. 

On the experimental part, it would be interesting to test our algorithms in more environments~\cite{melting_pot, social_jax} or realistic real-world simulation~\cite{rice_n} and create more diverse scenarii for CleanUp. With an apple regeneration rate of 0.05, the harvesting pressure exerted by only 7 agents is too low to create a significant imbalance between supply and demand. An pertinent scenario would be to increase the number of agents and decrease the regeneration rate to create a situation of resource scarcity, where demand is greater than supply.

Another way to improve our experiment in the Utilitarian setting would be to use the most recent state-of-the-art algorithms like HAPPO or HAA2C~\cite{HAML}.

\section{Conclusion}
This work introduces a novel framework that integrates Proportional Fairness with altruism to foster more equitable cooperation in multi-agent systems. We proposed a fair altruistic utility and extended it to sequential settings by defining a Fair Altruistic Markov Game, from which we derived novel fair policy gradient algorithms.

Our experiments in the 'CleanUp' social dilemma demonstrate that this approach not only achieves significantly fairer outcomes, as measured by the Gini index, but also leads to superior overall efficiency compared to the traditional utilitarian welfare objective with heterogeneous agents. By effectively balancing individual incentives with collective well-being, our method prevents agents from adopting highly specialized and inequitable roles, thus paving the way for more stable and sustainable cooperation.






\bibliographystyle{ACM-Reference-Format} 
\bibliography{AAMAS}

@article{selfishness_SR,
	title = {Selfishness {Level} {Induces} {Cooperation} in {Sequential} {Social} {Dilemmas}},
	language = {en},
	author = {Roesch, Stefan and Leonardos, Stefanos and Du, Yali},
}

@misc{simir_to_us_but_not_deep,
	title = {Achieving {Fairness} in {Multi}-{Agent} {Markov} {Decision} {Processes} {Using} {Reinforcement} {Learning}},
	url = {http://arxiv.org/abs/2306.00324},
	doi = {10.48550/arXiv.2306.00324},
	urldate = {2025-10-06},
	publisher = {arXiv},
	author = {Ju, Peizhong and Ghosh, Arnob and Shroff, Ness B.},
	month = jun,
	year = {2023},
	note = {arXiv:2306.00324 [cs]},
}

@book{ marl-book,
  author = {Stefano V. Albrecht and Filippos Christianos and Lukas Sch\"afer},
  title = {Multi-Agent Reinforcement Learning: Foundations and Modern Approaches},
  publisher = {MIT Press},
  year = {2024},
  url = {https://www.marl-book.com}
}

@article{nsw_def,
	title = {The {Nash} {Social} {Welfare} {Function}},
	volume = {47},
	issn = {00129682},
	url = {https://www.jstor.org/stable/1914191?origin=crossref},
	doi = {10.2307/1914191},
	language = {en},
	number = {2},
	urldate = {2025-10-06},
	journal = {Econometrica},
	author = {Kaneko, Mamoru and Nakamura, Kenjiro},
	month = mar,
	year = {1979},
	pages = {423},
}

@article{weighted_SR,
	title = {Towards {Fair} and {Equitable} {Policy} {Learning} in {Cooperative} {Multi}-{Agent} {Reinforcement} {Learning}},
	language = {en},
	author = {Siddique, Umer and Li, Peilang and Cao, Yongcan},
	year = {2024},
}

@misc{melting_pot,
	title = {Melting {Pot} 2.0},
	url = {http://arxiv.org/abs/2211.13746},
	doi = {10.48550/arXiv.2211.13746},
	urldate = {2025-10-06},
	publisher = {arXiv},
	author = {Agapiou, John P. and Vezhnevets, Alexander Sasha and Duéñez-Guzmán, Edgar A. and Matyas, Jayd and Mao, Yiran and Sunehag, Peter and Köster, Raphael and Madhushani, Udari and Kopparapu, Kavya and Comanescu, Ramona and Strouse, D. J. and Johanson, Michael B. and Singh, Sukhdeep and Haas, Julia and Mordatch, Igor and Mobbs, Dean and Leibo, Joel Z.},
	month = oct,
	year = {2023},
	note = {arXiv:2211.13746 [cs]},
}

@article{rice_n,
	title = {{AI} for {Global} {Climate} {Cooperation}: {Modeling} {Global} {Climate} {Negotiations}, {Agreements}, and {Long}-{Term} {Cooperation} in {RICE}-{N}},
	issn = {1556-5068},
	shorttitle = {{AI} for {Global} {Climate} {Cooperation}},
	url = {https://www.ssrn.com/abstract=4189735},
	doi = {10.2139/ssrn.4189735},
	language = {en},
	urldate = {2025-10-06},
	journal = {SSRN Electronic Journal},
	author = {Zhang, Tianyu and Williams, Andrew and Phade, Soham and Srinivasa, Sunil and Zhang, Yang and Gupta, Prateek and Bengio, Yoshua and Zheng, Stephan},
	year = {2022},
}

@misc{social_jax,
	title = {{SocialJax}: {An} {Evaluation} {Suite} for {Multi}-agent {Reinforcement} {Learning} in {Sequential} {Social} {Dilemmas}},
	shorttitle = {{SocialJax}},
	url = {http://arxiv.org/abs/2503.14576},
	doi = {10.48550/arXiv.2503.14576},
	urldate = {2025-10-06},
	publisher = {arXiv},
	author = {Guo, Zihao and Shi, Shuqing and Willis, Richard and Tomilin, Tristan and Leibo, Joel Z. and Du, Yali},
	month = may,
	year = {2025},
	note = {arXiv:2503.14576 [cs]},
}

@article{nsw,
	title = {The {Unreasonable} {Fairness} of {Maximum} {Nash} {Welfare}},
	volume = {7},
	issn = {2167-8375, 2167-8383},
	url = {https://dl.acm.org/doi/10.1145/3355902},
	doi = {10.1145/3355902},
	language = {en},
	number = {3},
	urldate = {2025-10-06},
	journal = {ACM Transactions on Economics and Computation},
	author = {Caragiannis, Ioannis and Kurokawa, David and Moulin, Hervé and Procaccia, Ariel D. and Shah, Nisarg and Wang, Junxing},
	month = aug,
	year = {2019},
	pages = {1--32},
}

@inproceedings{reward_shaping_0,
	title = {Learning {Fairness} in {Multi}-{Agent} {Systems}},
	volume = {32},
	url = {https://proceedings.neurips.cc/paper_files/paper/2019/hash/10493aa88605cad5ab4752b04a63d172-Abstract.html},
	urldate = {2025-10-06},
	booktitle = {Advances in {Neural} {Information} {Processing} {Systems}},
	publisher = {Curran Associates, Inc.},
	author = {Jiang, Jiechuan and Lu, Zongqing},
	year = {2019},
}

@misc{yu_surprising_2022,
	title = {The {Surprising} {Effectiveness} of {PPO} in {Cooperative}, {Multi}-{Agent} {Games}},
	url = {http://arxiv.org/abs/2103.01955},
	doi = {10.48550/arXiv.2103.01955},
	urldate = {2025-10-06},
	publisher = {arXiv},
	author = {Yu, Chao and Velu, Akash and Vinitsky, Eugene and Gao, Jiaxuan and Wang, Yu and Bayen, Alexandre and Wu, Yi},
	month = nov,
	year = {2022},
	note = {arXiv:2103.01955 [cs]},
}

@inproceedings{game_theory_harm_ratio,
	address = {San Luis Potosi Mexico},
	title = {Harm {Ratio}: {A} {Novel} and {Versatile} {Fairness} {Criterion}},
	isbn = {9798400712227},
	shorttitle = {Harm {Ratio}},
	url = {https://dl.acm.org/doi/10.1145/3689904.3694701},
	doi = {10.1145/3689904.3694701},
	language = {en},
	urldate = {2025-10-06},
	booktitle = {Proceedings of the 4th {ACM} {Conference} on {Equity} and {Access} in {Algorithms}, {Mechanisms}, and {Optimization}},
	publisher = {ACM},
	author = {Ebadian, Soroush and Freeman, Rupert and Shah, Nisarg},
	month = oct,
	year = {2024},
	pages = {1--14},
}

@incollection{alpha_altruistic,
	address = {Berlin, Heidelberg},
	title = {The {Robust} {Price} of {Anarchy} of {Altruistic} {Games}},
	volume = {7090},
	isbn = {978-3-642-25509-0 978-3-642-25510-6},
	url = {http://link.springer.com/10.1007/978-3-642-25510-6_33},
	language = {en},
	urldate = {2025-10-06},
	booktitle = {Internet and {Network} {Economics}},
	publisher = {Springer Berlin Heidelberg},
	author = {Chen, Po-An and De Keijzer, Bart and Kempe, David and Schäfer, Guido},
	editor = {Chen, Ning and Elkind, Edith and Koutsoupias, Elias},
	year = {2011},
	doi = {10.1007/978-3-642-25510-6_33},
	note = {Series Title: Lecture Notes in Computer Science},
	pages = {383--390},
}

@article{pf_intro,
	title = {Rate control for communication networks: shadow prices, proportional fairness and stability},
	language = {en},
	author = {Kelly, FP and Maulloo, AK and Tan, DKH},
}

@misc{reward_shaping_1,
	title = {Environmental-{Impact} {Based} {Multi}-{Agent} {Reinforcement} {Learning}},
	url = {http://arxiv.org/abs/2311.04240},
	doi = {10.48550/arXiv.2311.04240},
	urldate = {2025-10-08},
	publisher = {arXiv},
	author = {Alamiyan-Harandi, Farinaz and Ramazi, Pouria},
	month = nov,
	year = {2023},
	note = {arXiv:2311.04240 [cs]},
}

@misc{role_reward_shaping,
	title = {Role {Play}: {Learning} {Adaptive} {Role}-{Specific} {Strategies} in {Multi}-{Agent} {Interactions}},
	shorttitle = {Role {Play}},
	url = {http://arxiv.org/abs/2411.01166},
	doi = {10.48550/arXiv.2411.01166},
	urldate = {2025-10-08},
	publisher = {arXiv},
	author = {Long, Weifan and Wen, Wen and Zhai, Peng and Zhang, Lihua},
	month = nov,
	year = {2024},
	note = {arXiv:2411.01166 [cs]},
}

@inproceedings{reward_shaping_2,
	address = {Yokohama, Japan},
	title = {Balancing {Individual} {Preferences} and {Shared} {Objectives} in {Multiagent} {Reinforcement} {Learning}},
	isbn = {978-0-9992411-6-5},
	url = {https://www.ijcai.org/proceedings/2020/347},
	doi = {10.24963/ijcai.2020/347},
	language = {en},
	urldate = {2025-10-08},
	booktitle = {Proceedings of the {Twenty}-{Ninth} {International} {Joint} {Conference} on {Artificial} {Intelligence}},
	publisher = {International Joint Conferences on Artificial Intelligence Organization},
	author = {Durugkar, Ishan and Liebman, Elad and Stone, Peter},
	month = jul,
	year = {2020},
	pages = {2505--2511},
}

@misc{PF_study,
	title = {Proportionally {Fair} {Online} {Allocation} of {Public} {Goods} with {Predictions}},
	url = {http://arxiv.org/abs/2209.15305},
	doi = {10.48550/arXiv.2209.15305},
	urldate = {2025-10-08},
	publisher = {arXiv},
	author = {Banerjee, Siddhartha and Gkatzelis, Vasilis and Hossain, Safwan and Jin, Billy and Micha, Evi and Shah, Nisarg},
	month = sep,
	year = {2022},
	note = {arXiv:2209.15305 [cs]},
}

@misc{similar_to_us_but_no_policy_gradient,
	title = {Socially {Fair} {Reinforcement} {Learning}},
	url = {http://arxiv.org/abs/2208.12584},
	doi = {10.48550/arXiv.2208.12584},
	urldate = {2025-10-08},
	publisher = {arXiv},
	author = {Mandal, Debmalya and Gan, Jiarui},
	month = feb,
	year = {2023},
	note = {arXiv:2208.12584 [cs]},
}

@book{sutton_reinforcement_2014,
	address = {Cambridge, Massachusetts},
	edition = {Nachdruck},
	series = {Adaptive computation and machine learning},
	title = {Reinforcement learning: an introduction},
	isbn = {978-0-262-19398-6},
	shorttitle = {Reinforcement learning},
	language = {en},
	publisher = {The MIT Press},
	author = {Sutton, Richard S. and Barto, Andrew},
	year = {2014},
}

@misc{HAML,
	title = {Heterogeneous-{Agent} {Mirror} {Learning}: {A} {Continuum} of {Solutions} to {Cooperative} {MARL}},
	shorttitle = {Heterogeneous-{Agent} {Mirror} {Learning}},
	url = {http://arxiv.org/abs/2208.01682},
	doi = {10.48550/arXiv.2208.01682},
	urldate = {2025-10-08},
	publisher = {arXiv},
	author = {Kuba, Jakub Grudzien and Feng, Xidong and Ding, Shiyao and Dong, Hao and Wang, Jun and Yang, Yaodong},
	month = aug,
	year = {2022},
	note = {arXiv:2208.01682 [cs]},
}

@article{HARL,
	title = {Heterogeneous-{Agent} {Reinforcement} {Learning}},
	language = {en},
	author = {Zhong, Yifan and Kuba, Jakub Grudzien and Feng, Xidong and Hu, Siyi and Ji, Jiaming and Yang, Yaodong},
}

@misc{reward_shaping_3,
	title = {Skill-{Aligned} {Fairness} in {Multi}-{Agent} {Learning} for {Collaboration} in {Healthcare}},
	url = {http://arxiv.org/abs/2508.18708},
	doi = {10.48550/arXiv.2508.18708},
	urldate = {2025-10-08},
	publisher = {arXiv},
	author = {Ekpo, Promise Osaine and La, Brian and Wiener, Thomas and Agarwal, Saesha and Agrawal, Arshia and Gonzalez-Pumariega, Gonzalo and Molu, Lekan P. and Taylor, Angelique},
	month = sep,
	year = {2025},
	note = {arXiv:2508.18708 [cs]},
}

@misc{cooperation,
	title = {Multi-{Agent} {Actor}-{Critic} for {Mixed} {Cooperative}-{Competitive} {Environments}},
	url = {http://arxiv.org/abs/1706.02275},
	doi = {10.48550/arXiv.1706.02275},
	urldate = {2025-10-08},
	publisher = {arXiv},
	author = {Lowe, Ryan and Wu, Yi and Tamar, Aviv and Harb, Jean and Abbeel, Pieter and Mordatch, Igor},
	month = mar,
	year = {2020},
	note = {arXiv:1706.02275 [cs]},
}

@article{tragedy_of_common,
  title={The tragedy of the commons: the population problem has no technical solution; it requires a fundamental extension in morality.},
  author={Hardin, Garrett},
  journal={science},
  volume={162},
  number={3859},
  pages={1243--1248},
  year={1968},
  publisher={American Association for the Advancement of Science}
}

@article{social_dilemma,
author = {Dawes, Robyn M. and Messick, David M.},
title = {Social Dilemmas},
journal = {International Journal of Psychology},
volume = {35},
number = {2},
pages = {111-116},
doi = {https://doi.org/10.1080/002075900399402},
url = {https://onlinelibrary.wiley.com/doi/abs/10.1080/002075900399402},
eprint = {https://onlinelibrary.wiley.com/doi/pdf/10.1080/002075900399402},
year = {2000}
}

@misc{trpo,
      title={Trust Region Policy Optimization}, 
      author={John Schulman and Sergey Levine and Philipp Moritz and Michael I. Jordan and Pieter Abbeel},
      year={2017},
      eprint={1502.05477},
      archivePrefix={arXiv},
      primaryClass={cs.LG},
      url={https://arxiv.org/abs/1502.05477}, 
}

@misc{hughes2018inequityaversionimprovescooperation,
      title={Inequity aversion improves cooperation in intertemporal social dilemmas}, 
      author={Edward Hughes and Joel Z. Leibo and Matthew G. Phillips and Karl Tuyls and Edgar A. Duéñez-Guzmán and Antonio García Castañeda and Iain Dunning and Tina Zhu and Kevin R. McKee and Raphael Koster and Heather Roff and Thore Graepel},
      year={2018},
      eprint={1803.08884},
      archivePrefix={arXiv},
      primaryClass={cs.NE},
      url={https://arxiv.org/abs/1803.08884}, 
}

@book{10.7551/mitpress/2954.001.0001,
    author = {Moulin, Hervé},
    title = {Fair Division and Collective Welfare},
    publisher = {The MIT Press},
    year = {2003},
    month = {01},
    isbn = {9780262280297},
    doi = {10.7551/mitpress/2954.001.0001},
    url = {https://doi.org/10.7551/mitpress/2954.001.0001},
}


\clearpage
\onecolumn

\appendix

\section*{Appendix}
\section{Methodology}  \label{app:appendix}
\begin{theorem}
    The altruism level of social dilemma is 
    \begin{equation}
        \alpha_G=
        \begin{cases}
            0 & \text{if } T \le R, \\
            \frac{\log T-\log R}{\log R-\log S} & \text{if } T > R.
        \end{cases}
    \end{equation}
\end{theorem}

At the end of the proof, we need to ensure the consistencyness with the definition of the altruism level $\alpha$. Namely, the ratio $\frac{\log T-\log R}{\log R-\log S}$ should be lower than $1$
\begin{proof}
    Requiring the ratio $\frac{\log T-\log R}{\log R-\log S} < 1$ is equivalent to the constraint $TS \leq R^2$.

    Starting from the third inequality of social dilemmas, we have 
    \begin{align*}
        T+S&<2R,\\
        T^2+2TS + S^2 &<4R^2,\\
        \frac{T^2+S^2}{4}+\frac{1}{2}TS&<R^2.
    \end{align*}

    A well-known inequality for any real values $T$ and $S$ is
    \begin{equation*}
        T^2 +S^2 \geq 2TS.
    \end{equation*}
    
    Then, we obtain the following inequality 
    \begin{align*}
        \frac{T^2+S^2}{4}+\frac{1}{2}TS&<R^2, \\
        \frac{2\cdot TS}{4} + \frac{1}{2}TS&\leq R^2,\\
        TS\leq R^2.
    \end{align*}
\end{proof}

\section{FAIR ALTRUISTIC MARKOV GAME}
Here, we provide a more complete proof for policy gradient and advantage policy gradient theorems.
\begin{proof}
Here, $V_j^{\vec{\theta}}$ is assumed to be differentiable with respect to any $\theta_i$.

Let's derive the objective for agent $j$,
\begin{align*}
    \nabla_{\theta_i}J_i(\vec{\theta})&=\nabla_{\theta_i}\underset{s_0\sim\rho_0}{\mathbb{E}}\left[\sum_{j=1}^N c_i(j)\log V_j^{\vec{\theta}}(s_0)\right] \\
    &=\underset{s_0\sim\rho_0}{\mathbb{E}}\left[\sum_{j=1}^N c_i(j)\nabla_{\theta_i}\left[\log V_j^{\vec{\theta}}(s_0)\right]\right] \\
    &=\underset{s_0\sim\rho_0}{\mathbb{E}}\left[\sum_{j=1}^N c_i(j)\frac{\nabla_{\theta_i}V_j^{\vec{\theta}}(s_0)}{V_j^{\vec{\theta}}(s_0)}\right]&& \text{(using the chain rule).} \\
\end{align*}

Then, we need to find the gradient $\nabla_{\theta_i}V_j^{\vec{\theta}}(s_0)$. To do so, we start from the Bellman equation, differentiate it. Then, we define a contraction mapping on the state value function space and apply the Banach-Picard Fixed-Point Theorem to get the formula of the gradient.

\textbf{1. Differentiate the Bellman equation.}

Let's recall the Bellman equation for the state-value function :
\begin{equation*}
    V_j^{\vec{\theta}}(s) = \sum_{\vec{a} \in \mathcal{A}} \pi_{\vec{\theta}}(\vec{a}\,|\, s) \left( r_j(s, \vec{a}) + \gamma \sum_{s' \in \mathcal{S}} \mathbb{P}(s' | s, \vec{a}) V_j^{\vec{\theta}}(s') \right)
\end{equation*}
Note that $r_i$, $P$, and policies $\pi_{\theta_j}$ for $j \neq i$ do not depend on $\theta_i$.
Then, the Bellman equation can be differentiate using the product rule :
\begin{align*}
    \nabla_{\theta_i}V_j^{\vec{\theta}}(s) 
    &= \sum_{\vec{a} \in \mathcal{A}} \nabla_{\theta_i}\pi_{\vec{\theta}}(\vec{a}\,|\, s) \cdot\left( r_j(s, \vec{a}) + \gamma \sum_{s' \in \mathcal{S}} \mathbb{P}(s' | s, \vec{a}) V_j^{\vec{\theta}}(s') \right) \\
    &\quad\quad +\,  \pi_{\vec{\theta}}(\vec{a}\,|\, s)\cdot\nabla_{\theta_i} \left( r_j(s, \vec{a}) + \gamma \sum_{s' \in \mathcal{S}} \mathbb{P}(s' | s, \vec{a}) V_j^{\vec{\theta}}(s') \right)\\
    &=\sum_{\vec{a} \in \mathcal{A}} \nabla_{\theta_i}\pi_{\vec{\theta}}(\vec{a}\,|\, s) \cdot\left( r_j(s, \vec{a}) + \gamma \sum_{s' \in \mathcal{S}} \mathbb{P}(s' | s, \vec{a}) V_j^{\vec{\theta}}(s') \right) \\
    &\quad\quad +\,  \pi_{\vec{\theta}}(\vec{a}\,|\, s)\cdot \gamma \sum_{s' \in \mathcal{S}} \mathbb{P}(s' | s, \vec{a}) \nabla_{\theta_i} V_j^{\vec{\theta}}(s')\\
\end{align*}

Let's analyze the first term in the sum.

\textbf{Term 1}: Unrolling the joint policy 
$$\nabla_{\theta_i} \pi_{\vec{\theta}}(\vec{a}|s) = \nabla_{\theta_i} \left( \pi_{\theta_i}(a_i|s) \prod_{k \neq i} \pi_{\theta_k}(a_k|s) \right) = \left( \nabla_{\theta_i} \pi_{\theta_i}(a_i|s) \right) \prod_{k \neq i} \pi_k(a_k|s)$$
Then, using the log-derivative trick, $\nabla f = f \cdot\nabla \log f$:
$$\nabla_{\theta_i} \pi_{\vec{\theta}}(\vec{a}|s)= \left( \pi_{\theta_i}(a_i|s) \nabla_{\theta_i} \log \pi_{\theta_i}(a_i|s) \right) \prod_{k \neq i} \pi_{\theta_k}(a_k|s) = \pi_{\vec{\theta}}(\vec{a}|s)\cdot \nabla_{\theta_i} \log \pi_{\theta_i}(a_i|s)$$
Also, recall the definition of the action-state-value function :
\begin{equation*}
     Q_j^{\vec{\theta}}(s,\, \vec{a})=r_j(s, \vec{a}) + \gamma \sum_{s' \in \mathcal{S}} \mathbb{P}(s' | s, \vec{a}) V_j^{\vec{\theta}}(s').
\end{equation*}
Therefore, the first term can be rewritten 
\begin{equation*}
    \pi_{\vec{\theta}}(\vec{a}|s)\cdot \nabla_{\theta_i} \log \pi_{\theta_i}(a_i|s)\cdot Q_j^{\vec{\theta}}(s,\, \vec{a})
\end{equation*}

Now, substitute this back into the expression for $\nabla_{\theta_i} V_j^{\vec{\theta}}(s)$:
$$\nabla_{\theta_i} V_j^{\vec{\theta}}(s) = \sum_{\vec{a}} \pi_{\vec{\theta}}(\vec{a}|s) \left[ (\nabla_{\theta_i} \log \pi_{\theta_i}(a_i|s)) Q_j^{\vec{\theta}}(s, \vec{a}) + \gamma \sum_{s'} \mathbb{P}(s'|s, \vec{a}) \nabla_{\theta_i} V_j^{\vec{\theta}}(s') \right]$$

Let's define a function $G_{i, j}^{\vec{\theta}}(s)$ which represents the immediate gradient component at state $s$:
$$G_{i, j}^{\vec{\theta}}(s) := \sum_{\vec{a}\in\mathcal{A}} \pi_{\vec{\theta}}(\vec{a}|s) (\nabla_{\theta_i} \log \pi_{\theta_i}(a_i|s)) Q_j^{\vec{\theta}}(s, \vec{a}) = \mathbb{E}_{\vec{a} \sim \vec{\pi}_\theta} [ \nabla_{\theta_i} \log \pi_{\theta_i}(a_i|s) Q_j^{\vec{\theta}}(s, \vec{a}) ]$$

With this, our recursive equation becomes:
$$\nabla_{\theta_i} V_j^{\vec{\theta}}(s) = G_{i,j}^{\vec{\theta}}(s) + \gamma \sum_{\vec{a}\in\mathcal{A}} \pi_{\vec{\theta}}(\vec{a}|s) \sum_{s'} \mathbb{P}(s'|s, \vec{a}) \nabla_{\theta_i} V_j^{\vec{\theta}}(s')$$

\textbf{2. Define Bellman operator for the value function gradient}

This equation defines a fixed-point relationship. Let $g_j: S \to \mathbb{R}^{|\theta_j|}$ be a vector field of gradients. We can define the Bellman operator for the value function gradient, $\mathfrak{T}_{i, j}^{\vec{\theta}}$, as:
$$(\mathfrak{T}_{i, j}^{\vec{\pi}} g_i)(s) := G_{i,j}^{\vec{\theta}}(s) + \gamma \mathbb{E}_{\vec{a} \sim \pi_{\vec{\theta}}(\cdot|s), s' \sim P(\cdot|s,\vec{a})} [g_j(s')]$$

The gradient we are looking for, $\nabla_{\theta_i} V_j^{\vec{\theta}}$, is a fixed point of this operator:
$$\nabla_{\theta_i} V_j^{\vec{\theta}} = \mathfrak{T}_{i, j}^{\vec{\theta}} (\nabla_{\theta_i} V_j^{\vec{\theta}})$$

\textbf{3. Application of the Banach-Picard Fixed-Point Theorem}

To prove that this operator has a unique fixed point, we need to show it is a contraction mapping in a complete metric space. Let's consider the space of all gradient vector fields $g_j: S \to \mathbb{R}^{|\theta_i|}$ equipped with the infinity norm:
$$\|g_j\|_{\infty} = \max_{s \in S} \|g_j(s)\|_2$$
where $\|\cdot\|_2$ is the standard Euclidean norm.

Let $g_1$ and $g_2$ be two gradient vector fields in this space. We examine the distance between their mappings under the operator:
\begin{align*}
    (\mathfrak{T}_{i, j}^{\vec{\theta}} g_1)(s) - (\mathfrak{T}_{i, j}^{\vec{\theta}} g_2)(s) &= \left( G_{i,j}^{\vec{\theta}}(s) + \gamma \mathbb{E}_{s', \vec{a}} [g_1(s')] \right) - \left( G_{i,j}^{\vec{\theta}}(s) + \gamma \mathbb{E}_{s', \vec{a}} [g_2(s')] \right) \\
    &= \gamma \mathbb{E}_{\vec{a} \sim \vec{\pi}(\cdot|s), s' \sim P(\cdot|s,\vec{a})} [g_1(s') - g_2(s')]
\end{align*}

Now, we take the norm:
$$\| (\mathfrak{T}_{i, j}^{\vec{\theta}} g_1)(s) - (\mathfrak{T}_{i, j}^{\vec{\theta}} g_2)(s) \|_2 = \gamma \left\| \sum_{s'\in\mathcal{S}} \left( \sum_{\vec{a}\in\mathcal{A}} \pi_{\vec{\theta}}(\vec{a}|s) \mathbb{P}(s'|s, \vec{a}) \right) (g_1(s') - g_2(s')) \right\|_2$$
Using the triangle inequality and properties of norms:
$$\le \gamma \sum_{s'\in\mathcal{S}} \left( \sum_{\vec{a}\in\mathcal{A}} \pi_{\vec{\theta}}(\vec{a}|s) \mathbb{P}(s'|s, \vec{a}) \right) \| g_1(s') - g_2(s') \|_2$$
The term in the parenthesis is the total probability of transitioning from $s$ to $s'$ under the joint policy $\pi_{\vec{\theta}}$, let's call it $p_{\vec{\theta}}(s'|s)$.
$$\le \gamma \sum_{s'} p_{\vec{\pi}}(s'|s) \| g_1(s') - g_2(s') \|_2$$
By definition of the infinity norm, $\| g_1(s') - g_2(s') \|_2 \le \| g_1 - g_2 \|_{\infty}$ for any $s'$.
$$\le \gamma \sum_{s'} P^{\vec{\pi}}(s'|s) \| g_1 - g_2 \|_{\infty}$$
Since probabilities sum to one, $\sum_{s'} p_{\vec{\theta}}(s'|s) = 1$:
$$\| (\mathfrak{T}_{i, j}^{\vec{\theta}} g_1)(s) - (\mathfrak{T}_{i, j}^{\vec{\theta}} g_2)(s) \|_2 \le \gamma \| g_1 - g_2 \|_{\infty}.$$

This inequality holds for all states $s \in \mathcal{S}$. Therefore, taking the maximum over all $s$:
$$\| \mathfrak{T}_{i, j}^{\vec{\theta}} g_1 - \mathfrak{T}_{i, j}^{\vec{\theta}} g_2 \|_{\infty} \le \gamma \| g_1 - g_2 \|_{\infty}$$

Since $\gamma \in [0, 1)$, the operator $\mathfrak{T}_{i, j}^{\vec{\theta}}$ is a contraction mapping. By the \textbf{Banach-Picard fixed-point theorem}, it has a unique fixed point. This proves that the gradient vector field $\nabla_{\theta_i} V_j^{\vec{\theta}}$ is unique.

\hfill

\textbf{4. Policy Gradient Theorem for Markov Games}

The fixed-point equation $\nabla_{\theta_i} V_j^{\vec{\theta}} = \mathfrak{T}_{i, j}^{\vec{\theta}} (\nabla_{\theta_i} V_j^{\vec{\theta}})$ gives us the final theorem. The \textbf{Banach-Picard fixed-point theorem} also states that for any sequence $g_{n+1}=\mathfrak{T}_{i, j}^{\vec{\theta}}g_n$ and $g_0:\mathcal{S}\rightarrow\mathbb{R}^{|\theta_j|}$, $(g_n)_{n\in\mathbb{N}}$ converges $\nabla_{\theta_i}V_j^{\vec{\theta}}$.

We can find the solution by unrolling the recursion by starting with $g_0=0$:
\begin{align*}
    g_1(s_0)&=(\mathfrak{T}_{i, j}^{\vec{\theta}}g_0)(s_0)=G_{i,j}^{\vec{\theta}}(s_0) \\
    g_2(s_0)&=(\mathfrak{T}_{i, j}^{\vec{\theta}}g_1)(s_0)=G_{i,j}^{\vec{\theta}}(s_0) +\gamma\underset{\substack{\vec{a}_0\sim \pi_{\vec{\theta}}(.|s_0) \\ s_1\sim P(.|s_0, \vec{a}_0)}}{\mathbb{E}} \left[G_{i,j}^{\vec{\theta}}(s_1)\right]\\
    g_3(s_0)&=(\mathfrak{T}_{i, j}^{\vec{\theta}}g_2)(s_0)=G_{i,j}^{\vec{\theta}}(s_0) +\gamma\underset{\substack{\vec{a}_0\sim \pi_{\vec{\theta}}(.|s_0) \\ s_1\sim P(.|s_0, \vec{a}_0)}}{\mathbb{E}} \left[G_{i,j}^{\vec{\theta}}(s_1) +\gamma\underset{\substack{\vec{a}_1\sim \pi_{\vec{\theta}}(.|s_1) \\ s_2\sim P(.|s_1, \vec{a}_1)}}{\mathbb{E}} \left[G_{i,j}^{\vec{\theta}}(s_2)\right]\right]\\
    &\vdots\\
    g_n(s_0)&=(\mathfrak{T}_{i, j}^{\vec{\theta}}g_{n-1})(s_0)=\underset{\tau\sim\mathbb{P}^{\vec{\theta}}}{\mathbb{E}}\left[\sum_{t=0}^{n-1}\gamma^t G_{i,j}^{\vec{\theta}}(s_t)\;\middle|\;s_0\right].
\end{align*}

Therefore, as $g_n$ converges to $\nabla_{\theta_i} V_j^{\vec{\theta}}$, the following equality holds 
\begin{equation*}
    \nabla_{\theta_i} V_j^{\vec{\theta}}(s_0)=g_\infty(s_0)=\underset{\tau\sim \mathbb{P}^{\vec{\theta}}}{\mathbb{E}} \left[ \sum_{t=0}^{\infty} \gamma^t G_{i,j}^{\vec{\theta}}(s_t) \;\middle|\; s_0\right]
\end{equation*}

Substituting the definition of $G_{i,j}^{\vec{\theta}}(s_t)$:
$$\nabla_{\theta_i} V_j^{\vec{\theta}}(s_0) = \underset{\tau\sim \mathbb{P}^{\vec{\theta}}}{\mathbb{E}} \left[ \sum_{t=0}^{\infty} \gamma^t \underset{\vec{a}_t \sim \pi_{\vec{\theta}}(\cdot|s_t)}{\mathbb{E}} \left[ \nabla_{\theta_i} \log \pi_{\theta_i}(a_{i,t}|s_t) Q_j^{\vec{\theta}}(s_t, \vec{a}_t) \right]\;\middle|\; s_0 \right]$$
The outer expectation already accounts for sampling actions, so we can simplify this to:

$$\nabla_{\theta_i} V_j^{\vec{\pi}}(s_0) = \underset{\tau\sim \mathbb{P}^{\vec{\theta}}}{\mathbb{E}}\left[ \sum_{t=0}^{\infty} \gamma^t \nabla_{\theta_i} \log \pi_{\theta_i}(a_{i,t}|s_t) Q_j^{\vec{\theta}}(s_t, \vec{a}_t) \;\middle|\; s_0\right]$$

By injecting the expression of $\nabla_{\theta_i} V_j^{\vec{\pi}}(s_0)$ into the former gradient, we get the \textbf{Fair Policy Gradient Theorem}
\begin{align*}
    \nabla_{\theta_i}J_i(\vec{\theta})
    &=\underset{s_0\sim\rho_0}{\mathbb{E}}\left[\sum_{j=1}^N c_i(j)\frac{\nabla_{\theta_i}V_j^{\vec{\theta}}(s_0)}{V_j^{\vec{\theta}}(s_0)}\right] \\
    &=\underset{s_0\sim\rho_0}{\mathbb{E}}\left[\sum_{j=1}^N c_i(j)\frac{\underset{\tau\sim \mathbb{P}^{\vec{\theta}}}{\mathbb{E}}\left[ \sum_{t=0}^{\infty} \gamma^t \nabla_{\theta_i} \log \pi_{\theta_i}(a_{i,t}|s_t) Q_j^{\vec{\theta}}(s_t, \vec{a}_t) \;\middle|\; s_0\right]}{V_j^{\vec{\theta}}(s_0)}\right] \\
    &=\underset{s_0\sim\rho_0}{\mathbb{E}}\left[\sum_{j=1}^N c_i(j)\underset{\tau\sim \mathbb{P}^{\vec{\theta}}}{\mathbb{E}}\left[ \sum_{t=0}^{\infty} \gamma^t \nabla_{\theta_i} \log \pi_{\theta_i}(a_{i,t}|s_t) \frac{Q_j^{\vec{\theta}}(s_t, \vec{a}_t)}{V_j^{\vec{\theta}}(s_0)} \;\middle|\; s_0\right]\right] \\
    &=\underset{s_0\sim\rho_0}{\mathbb{E}}\left[\underset{\tau\sim \mathbb{P}^{\vec{\theta}}}{\mathbb{E}}\left[ \sum_{t=0}^{\infty} \gamma^t \nabla_{\theta_i} \log \pi_{\theta_i}(a_{i,t}|s_t) \left(\sum_{j=1}^N c_i(j)\frac{Q_j^{\vec{\theta}}(s_t, \vec{a}_t)}{V_j^{\vec{\theta}}(s_0)}\right)\;\middle|\; s_0 \right]\right] \\
    &=\underset{\tau\sim \mathbb{P}^{\vec{\theta}}}{\mathbb{E}}\left[ \sum_{t=0}^{\infty} \gamma^t \nabla_{\theta_i} \log \pi_{\theta_i}(a_{i,t}|s_t) \left(\sum_{j=1}^N c_i(j)\frac{Q_j^{\vec{\theta}}(s_t, \vec{a}_t)}{V_j^{\vec{\theta}}(s_0)}\right)\right]. \\
\end{align*}

\end{proof}

\begin{proof}
The proof relies on showing that the expected value of the gradient term associated with the baseline $V_i^{\vec{\pi}}(s_t)$ is zero. Let's consider the expectation of this baseline term at a single timestep $t$ for a given state $s_t$. The expectation is over the joint action $\vec{a}_t$ drawn from the joint policy $\vec{\pi}(\cdot|s_t)$.

We want to show that for any $f:\mathcal{S}\rightarrow\mathbb{R}$
$$\mathbb{E}_{\vec{a}_t \sim \pi_{\vec{\theta}}(\cdot|s_t)} \left[ \nabla_{\theta_i} \log \pi_i(a_{i,t}|s_t) f(s_t) \right] = 0$$

Let's expand the expectation:
$$\sum_{\vec{a}_t \in \mathcal{A}} \pi_{\vec{\theta}}(\vec{a}_t|s_t) \nabla_{\theta_i} \log \pi_i(a_{i,t}|s_t)f(s_t)$$

The function $f(s_t)$ depends on the state $s_t$ but not on the specific action $\vec{a}_t$ being sampled at that state, so we can pull it out of the sum over actions for which it is constant. More formally, we can separate the sum over agent $i$'s actions from the other agents' actions. Let $\vec{a}_t = (a_{i,t}, a_{-i,t})$, where $a_{-i,t}$ are the actions of all agents other than $i$.

$$\sum_{a_{-i,t}} \prod_{j \neq i}\pi_j(a_{j,t}|s_t) \sum_{a_{i,t}} \pi_i(a_{i,t}|s_t) \nabla_{\theta_i} \log \pi_i(a_{i,t}|s_t)f(s_t)$$

Since $f(s_t)$ does not depend on $a_{i,t}$, we can factor it out of the inner sum:
$$= \left( \sum_{a_{-i,t}} \prod_{j \neq i}\pi_{\theta_j}(a_{j,t}|s_t) \right)f(s_t) \left( \sum_{a_{i,t}} \pi_{\theta_i}(a_{i,t}|s_t) \nabla_{\theta_i} \log \pi_{\theta_i}(a_{i,t}|s_t) \right)$$

The first term is a product of probabilities summed over all actions, which equals 1. Let's focus on the last term. Using the log-derivative trick 
\begin{align*}
    \sum_{a_{i,t}} \pi_{\theta_i}(a_{i,t}|s_t) \nabla_{\theta_i} \log \pi_{\theta_i}(a_{i,t}|s_t) 
    &= \sum_{a_{i,t}} \nabla_{\theta_i} \pi_{\theta_i}(a_{i,t}|s_t)\\
    &= \nabla_{\theta_i} \sum_{a_{i,t}} \pi_{\theta_i}(a_{i,t}|s_t)&&\text{(sum of probabilities over all possible actions equals 1)}\\
    &= \nabla_{\theta_i} (1)\\
    &=0
\end{align*}

Since the expected value of the baseline term is zero, subtracting it from the term inside the expectation in the original policy gradient theorem does not change the total expectation. In this work, we use the baseline $f(s_t)=\sum_{j=1}^Nc_i(j)\frac{V_j^{\vec{\theta}}(s_t)}{V_j^{\vec{\theta}}(s_0)}$

Since the expectation of the baseline term is zero, we can subtract it inside the outer expectation:
\begin{align*}
    \nabla_{\theta_i}J_i(\vec{\theta})
    &=\underset{\tau\sim\mathbb{P}^{\vec{\theta}}}{\mathbb{E}}\left[\sum_{t=0}^\infty \sum_{j=1}^Nc_i(j)\frac{Q^{\vec{\theta}}_j(s_t, \, \vec{a}_t)}{V^{\vec{\theta}}_j(s_0)}\nabla_{\theta_i}\log\pi_{\theta_i}(a_{i,\, t} \;|\; s_t)\right] \\
    &=\underset{\tau\sim\mathbb{P}^{\vec{\theta}}}{\mathbb{E}}\left[\sum_{t=0}^\infty \sum_{j=1}^Nc_i(j)\frac{Q^{\vec{\theta}}_j(s_t, \, \vec{a}_t)}{V^{\vec{\theta}}_j(s_0)}\nabla_{\theta_i}\log\pi_{\theta_i}(a_{i,\, t} \;|\; s_t)\right]\\
    &\quad\quad-\underset{\tau\sim\mathbb{P}^{\vec{\theta}}}{\mathbb{E}}\left[\sum_{t=0}^\infty \left(\sum_{j=1}^Nc_i(j)\frac{V_j^{\vec{\theta}}(s_t)}{V_j^{\vec{\theta}}(s_0)}\right)\nabla_{\theta_i}\log\pi_{\theta_i}(a_{i,\, t} \;|\; s_t)\right]\\
    &=\underset{\tau\sim\mathbb{P}^{\vec{\theta}}}{\mathbb{E}}\left[\sum_{t=0}^\infty \sum_{j=1}^Nc_i(j)\frac{Q^{\vec{\theta}}_j(s_t, \, \vec{a}_t)-V^{\vec{\theta}}_j(s_t)}{V^{\vec{\theta}}_j(s_0)}\nabla_{\theta_i}\log\pi_{\theta_i}(a_{i,\, t} \;|\; s_t)\right] \\
    &=\underset{\tau\sim\mathbb{P}^{\vec{\theta}}}{\mathbb{E}}\left[\sum_{t=0}^\infty A^{F,\vec{\theta}}_i(s_t,\vec{a}_t, s_0)\nabla_{\theta_i}\log\pi_{\theta_i}(a_{i,\, t} \;|\; s_t)\right] \\
\end{align*}

\end{proof}

\section{Pseudo code examples}
\begin{algorithm}
\caption{Fair MAA2C}
\label{alg:ma-a2c}
\begin{algorithmic}[1]
\REQUIRE Global learning rate $lr$, discount factor $\gamma$, number of agents $N$, number of steps $T_{max}$, altruism level $\alpha$.
\ENSURE For each agent $i \in \{1, ..., N\}$, a critic network $V(\mathbf{s}; \phi_i)$ with random parameters $\phi_i$.
\ENSURE For each agent $i \in \{1, ..., N\}$, an actor network $\pi(a_i | o_i; \theta_i)$ with random parameters $\theta_i$.
\STATE Initialize global step counter $T \leftarrow 0$.

\WHILE{$T < T_{max}$}
    \STATE Reset gradients: $d\phi_i \leftarrow 0$, and $d\theta_i \leftarrow 0$ for all $i=1, ..., N$.
    \STATE Initialize environment and get initial joint observation $\mathbf{o} = \{o_1, ..., o_N\}$ and global state $\mathbf{s}$.
    \STATE Initialize empty episode buffer $D$.
    
    \FOR{$t = 1$ to a terminal state or max episode length}
        \STATE For each agent $i$, select action $a_i \sim \pi(a_i | o_i; \theta_i)$.
        \STATE Execute joint action $\mathbf{a} = \{a_1, ..., a_N\}$, observe rewards $\mathbf{r} = \{r_1, ..., r_N\}$, next joint observation $\mathbf{o}'$, and next global state $\mathbf{s}'$.
        \STATE Store $(\mathbf{s}, \mathbf{a}, \mathbf{r}, \mathbf{s}')$ in the episode buffer $D$.
        \STATE $\mathbf{s} \leftarrow \mathbf{s}'$, $\mathbf{o} \leftarrow \mathbf{o}'$.
        \STATE $T \leftarrow T+1$.
    \ENDFOR

    \STATE For each agent $i$, estimate the return $R_{i,\, t}$ using TD update,
    \STATE For each agent $i$, estimate the state-value function of the initial state : $V_{i,\,0}\leftarrow V(s_0;\phi_i)$.
    \STATE For each agent $i$ and every timestep $t$, estimate the advantage function of agent $i$ (using GAE): $A_{i,\, t}$.
    \STATE For each agent $i$ and every timestep $t$, estimate the fair advantage function of agent $i$: $A^F_{i,\, t}\leftarrow \sum_j c_i(j)A_{j,\, t}/V_{j,\, 0}$.
    \STATE For each agent $i$, calculate the actor loss: $L_{\text{actor},\, i}(\theta_i)=\mathbb{E}_t[A^F_{i,\, t}\log\pi(a_i | o_i; \theta_i)]$.
    \STATE For each agent $i$,  calculate the critic loss: $L_{\text{critic},\, i}(\phi_i)=\mathbb{E}_t[(V(s_t;\phi_i)-R_t)^2]$
    \STATE For each agent $i$, update critic parameters: $\phi_i \leftarrow \phi_i - lr \cdot d\phi_i$.
    \STATE For each agent $i$, update actor parameters: $\theta_i \leftarrow \theta_i - \alpha \cdot d\theta_i$.

\ENDWHILE
\end{algorithmic}
\end{algorithm}

\begin{algorithm}
\caption{Fair MAPPO}
\label{alg:ma-ppo}
\begin{algorithmic}[1]
\REQUIRE Global learning rate $lr$, discount factor $\gamma$, number of agents $N$, number of steps $T_{max}$, altruism level $\alpha$, entropy regularization $\beta$, number of PPO update $T_{\text{PPO}}$.
\ENSURE For each agent $i \in \{1, ..., N\}$, a critic network $V(\mathbf{s}; \phi_i)$ with random parameters $\phi_i$.
\ENSURE For each agent $i \in \{1, ..., N\}$, an actor network $\pi(a_i | o_i; \theta_i)$ with random parameters $\theta_i$.
\STATE Initialize global step counter $T \leftarrow 0$.

\WHILE{$T < T_{max}$}
    \STATE Reset gradients: $d\phi_i \leftarrow 0$, and $d\theta_i \leftarrow 0$ for all $i=1, ..., N$.
    \STATE Initialize environment and get initial joint observation $\mathbf{o} = \{o_1, ..., o_N\}$ and global state $\mathbf{s}$.
    \STATE Initialize empty episode buffer $D$.
    
    \FOR{$t = 1$ to a terminal state or max episode length}
        \STATE For each agent $i$, select action $a_i \sim \pi(a_i | o_i; \theta_i)$.
        \STATE Execute joint action $\mathbf{a} = \{a_1, ..., a_N\}$, observe rewards $\mathbf{r} = \{r_1, ..., r_N\}$, next joint observation $\mathbf{o}'$, and next global state $\mathbf{s}'$.
        \STATE Store $(\mathbf{s}, \mathbf{a}, \mathbf{r}, \mathbf{s}')$ in the episode buffer $D$.
        \STATE $\mathbf{s} \leftarrow \mathbf{s}'$, $\mathbf{o} \leftarrow \mathbf{o}'$.
        \STATE $T \leftarrow T+1$.
    \ENDFOR

    \STATE For each agent $i\in\{1,\dots, N\}$, estimate the return $R_{i,\, t}$ using TD update,
    \STATE For each agent $i\in\{1,\dots, N\}$, estimate the state-value function of the initial state : $V_{i,\,0}\leftarrow V(s_0;\phi_i)$.
    \STATE For each agent $i\in\{1,\dots, N\}$ and every timestep $t$, estimate the advantage function of agent $i$ (using GAE): $A_{i,\, t}$.
    \STATE For each agent $i\in\{1,\dots, N\}$ and every timestep $t$, estimate the fair advantage function of agent $i$: $A^F_{i,\, t}\leftarrow \sum_j c_i(j)A_{j,\, t}/V_{j,\, 0}$.
    
    \FOR{$b=1$ to $T_{\text{PPO}}$}
        \STATE For each agent $i\in\{1,\dots, N\}$, calculate the actor PPO loss: $L_{\text{actor},\, i}(\theta_i)$.
        \STATE For each agent $i\in\{1,\dots, N\}$,  calculate the critic PPO loss: $L_{\text{critic},\, i}$
        \STATE For each agent $i$, update critic parameters: $\phi_i \leftarrow \phi_i - lr \cdot d\phi_i$.
        \STATE For each agent $i$, update actor parameters: $\theta_i \leftarrow \theta_i - \alpha \cdot d\theta_i$.
    \ENDFOR

\ENDWHILE
\end{algorithmic}
\end{algorithm}

\end{document}